\documentclass[a4paper,11pt]{article}

\newcommand{\typeof}{1} %

\usepackage{ifthen}
\newcommand{\condinc}[2]{\ifthenelse{\equal{\typeof}{0}}{#1}{#2}}

\usepackage{mathrsfs}

\usepackage{times}

\usepackage{graphics}
\usepackage{xy}

\usepackage{pdfsync} 

\usepackage{mathrsfs}
\condinc{}{\usepackage{a4wide,url}}
\usepackage{cmll}
\usepackage{euscript}
\usepackage{amssymb,amsmath}

\newdimen\proofrulebreadth \proofrulebreadth=.05em
\newdimen\proofdotseparation \proofdotseparation=1.25ex
\newdimen\proofrulebaseline \proofrulebaseline=2ex
\newcount\proofdotnumber \proofdotnumber=3
\let\then\relax
\def\hfi{\hskip0pt plus.0001fil}
\mathchardef\squigto="3A3B
\newif\ifinsideprooftree\insideprooftreefalse
\newif\ifonleftofproofrule\onleftofproofrulefalse
\newif\ifproofdots\proofdotsfalse
\newif\ifdoubleproof\doubleprooffalse
\let\wereinproofbit\relax
\newdimen\shortenproofleft
\newdimen\shortenproofright
\newdimen\proofbelowshift
\newbox\proofabove
\newbox\proofbelow
\newbox\proofrulename
\def\shiftproofbelow{\let\next\relax\afterassignment\setshiftproofbelow\dimen0 }
\def\shiftproofbelowneg{\def\next{\multiply\dimen0 by-1 }%
\afterassignment\setshiftproofbelow\dimen0 }
\def\setshiftproofbelow{\next\proofbelowshift=\dimen0 }
\def\setproofrulebreadth{\proofrulebreadth}

\def\prooftree{
\ifnum  \lastpenalty=1
\then   \unpenalty
\else   \onleftofproofrulefalse
\fi
\ifonleftofproofrule
\else   \ifinsideprooftree
        \then   \hskip.5em plus1fil
        \fi
\fi
\bgroup
\setbox\proofbelow=\hbox{}\setbox\proofrulename=\hbox{}%
\let\justifies\proofover\let\leadsto\proofoverdots\let\Justifies\proofoverdbl
\let\using\proofusing\let\[\prooftree
\ifinsideprooftree\let\]\endprooftree\fi
\proofdotsfalse\doubleprooffalse
\let\thickness\setproofrulebreadth
\let\shiftright\shiftproofbelow \let\shift\shiftproofbelow
\let\shiftleft\shiftproofbelowneg
\let\ifwasinsideprooftree\ifinsideprooftree
\insideprooftreetrue
\setbox\proofabove=\hbox\bgroup$\displaystyle 
\let\wereinproofbit\prooftree
\shortenproofleft=0pt \shortenproofright=0pt \proofbelowshift=0pt
\onleftofproofruletrue\penalty1
}

\def\eproofbit{
\ifx    \wereinproofbit\prooftree
\then   \ifcase \lastpenalty
        \then   \shortenproofright=0pt  
        \or     \unpenalty\hfil         
        \or     \unpenalty\unskip       
        \else   \shortenproofright=0pt  
        \fi
\fi
\global\dimen0=\shortenproofleft
\global\dimen1=\shortenproofright
\global\dimen2=\proofrulebreadth
\global\dimen3=\proofbelowshift
\global\dimen4=\proofdotseparation
\global\count255=\proofdotnumber
$\egroup  
\shortenproofleft=\dimen0
\shortenproofright=\dimen1
\proofrulebreadth=\dimen2
\proofbelowshift=\dimen3
\proofdotseparation=\dimen4
\proofdotnumber=\count255
}

\def\proofover{
\eproofbit 
\setbox\proofbelow=\hbox\bgroup 
\let\wereinproofbit\proofover
$\displaystyle
}%
\def\proofoverdbl{
\eproofbit 
\doubleprooftrue
\setbox\proofbelow=\hbox\bgroup 
\let\wereinproofbit\proofoverdbl
$\displaystyle
}%
\def\proofoverdots{
\eproofbit 
\proofdotstrue
\setbox\proofbelow=\hbox\bgroup 
\let\wereinproofbit\proofoverdots
$\displaystyle
}%
\def\proofusing{
\eproofbit 
\setbox\proofrulename=\hbox\bgroup 
\let\wereinproofbit\proofusing
\kern0.3em$
}

\def\endprooftree{
\eproofbit 
  \dimen5 =0pt
\dimen0=\wd\proofabove \advance\dimen0-\shortenproofleft
\advance\dimen0-\shortenproofright
\dimen1=.5\dimen0 \advance\dimen1-.5\wd\proofbelow
\dimen4=\dimen1
\advance\dimen1\proofbelowshift \advance\dimen4-\proofbelowshift
\ifdim  \dimen1<0pt
\then   \advance\shortenproofleft\dimen1
        \advance\dimen0-\dimen1
        \dimen1=0pt
        \ifdim  \shortenproofleft<0pt
        \then   \setbox\proofabove=\hbox{%
                        \kern-\shortenproofleft\unhbox\proofabove}%
                \shortenproofleft=0pt
        \fi
\fi
\ifdim  \dimen4<0pt
\then   \advance\shortenproofright\dimen4
        \advance\dimen0-\dimen4
        \dimen4=0pt
\fi
\ifdim  \shortenproofright<\wd\proofrulename
\then   \shortenproofright=\wd\proofrulename
\fi
\dimen2=\shortenproofleft \advance\dimen2 by\dimen1
\dimen3=\shortenproofright\advance\dimen3 by\dimen4
\ifproofdots
\then
        \dimen6=\shortenproofleft \advance\dimen6 .5\dimen0
        \setbox1=\vbox to\proofdotseparation{\vss\hbox{$\cdot$}\vss}%
        \setbox0=\hbox{%
                \advance\dimen6-.5\wd1
                \kern\dimen6
                $\vcenter to\proofdotnumber\proofdotseparation
                        {\leaders\box1\vfill}$%
                \unhbox\proofrulename}%
\else   \dimen6=\fontdimen22\the\textfont2 
        \dimen7=\dimen6
        \advance\dimen6by.5\proofrulebreadth
        \advance\dimen7by-.5\proofrulebreadth
        \setbox0=\hbox{%
                \kern\shortenproofleft
                \ifdoubleproof
                \then   \hbox to\dimen0{%
                        $\mathsurround0pt\mathord=\mkern-6mu%
                        \cleaders\hbox{$\mkern-2mu=\mkern-2mu$}\hfill
                        \mkern-6mu\mathord=$}%
                \else   \vrule height\dimen6 depth-\dimen7 width\dimen0
                \fi
                \unhbox\proofrulename}%
        \ht0=\dimen6 \dp0=-\dimen7
\fi
\let\doll\relax
\ifwasinsideprooftree
\then   \let\VBOX\vbox
\else   \ifmmode\else$\let\doll=$\fi
        \let\VBOX\vcenter
\fi
\VBOX   {\baselineskip\proofrulebaseline \lineskip.2ex
        \expandafter\lineskiplimit\ifproofdots0ex\else-0.6ex\fi
        \hbox   spread\dimen5   {\hfi\unhbox\proofabove\hfi}%
        \hbox{\box0}%
        \hbox   {\kern\dimen2 \box\proofbelow}}\doll%
\global\dimen2=\dimen2
\global\dimen3=\dimen3
\egroup 
\ifonleftofproofrule
\then   \shortenproofleft=\dimen2
\fi
\shortenproofright=\dimen3
\onleftofproofrulefalse
\ifinsideprooftree
\then   \hskip.5em plus 1fil \penalty2
\fi
}

\condinc{}{\newcommand{\qed}{}}

\newcommand{\comment}[1]{}

\newcommand{\typet}[2]{#1:#2}
\newcommand{\jd}[2]{#1\vdash #2}

\newcommand{\permone}{\sigma}

\newcommand{\pair}[2]{\langle #1,#2\rangle}
\newcommand{\pairtens}[2]{#1\otimes #2}

\newcommand{\linmap}{\multimap}
\newcommand{\abstr}[2]{\lambda #1.#2}

\newcommand{\pof}{\triangleright}

\newcommand{\QL}{\mathsf{Q}\Lambda}
\newcommand{\varone}{x}
\newcommand{\vartwo}{y}
\newcommand{\varthree}{z}
\newcommand{\varfour}{w}
\newcommand{\patone}{\pi}
\newcommand{\csone}{C}
\newcommand{\termone}{M}
\newcommand{\termtwo}{N}
\newcommand{\termthree}{L}
\newcommand{\termfour}{Q}
\newcommand{\cbitone}{B}
\newcommand{\bitone}{b}

\newcommand{\conone}{\Gamma}
\newcommand{\contwo}{\Delta}

\newcommand{\emcon}{\cdot}
\newcommand{\typeone}{A}
\newcommand{\typetwo}{B}
\newcommand{\typethree}{C}
\newcommand{\typefour}{D}

\newcommand{\unopone}{\mathbf{U}}

\newcommand{\opone}{U}

\newcommand{\subst}[3]{#1\{#2/#3\}}
\newcommand{\eqterm}{\approx}
\newcommand{\eqtermnf}{\sim}
\newcommand{\uset}{\mathscr{U}}
\newcommand{\reduct}[1]{#1^{\Downarrow}}

\newcommand{\IAM}[1]{\mathsf{IAM}_{#1}}
\newcommand{\autom}[1]{\mathcal{A}_{#1}}
\newcommand{\automc}[2]{\mathcal{A}_{#1}^{#2}}
\newcommand{\states}[1]{\mathcal{S}_{#1}}
\newcommand{\trans}[1]{\rightarrow_{#1}}
\newcommand{\occone}{O}
\newcommand{\occtwo}{P}
\newcommand{\occthree}{R}
\newcommand{\soccone}{\varphi}
\newcommand{\socctwo}{\psi}

\newcommand{\ctone}{C}
\newcommand{\cttwo}{D}
\newcommand{\emct}{[\cdot]}
\newcommand{\pconone}{P}
\newcommand{\pcontwo}{Q}
\newcommand{\nconone}{N}
\newcommand{\ncontwo}{M}
\newcommand{\qrone}{\mathbf{Q}}
\newcommand{\qrtwo}{\mathbf{R}}
\newcommand{\qrthree}{\mathbf{S}}

\newcommand{\poccs}[1]{\mathcal{P}(#1)}
\newcommand{\noccs}[1]{\mathcal{N}(#1)}
\newcommand{\emseq}{\varepsilon}

\newcommand{\tdone}{\pi}
\newcommand{\tdtwo}{\rho}
\newcommand{\tdthree}{\sigma}

\newcommand{\sutone}{\mathcal{T}}
\newcommand{\suttwo}{\mathcal{S}}
\newcommand{\sutthree}{\mathcal{V}}
\newcommand{\sutfour}{\mathcal{W}}
\newcommand{\sutfive}{\mathcal{X}}
\newcommand{\sutsix}{\mathcal{Y}}

\newcommand{\mapone}[1]{( #1)^{\bullet}}

\newcommand{\qcone}{C}

\newcommand{\MLL}{\textsf{MLL}}
\newcommand{\QMLL}{\textsf{QMLL}}
\newcommand{\atomone}{\alpha}
\newcommand{\atomtwo}{\beta}
\newcommand{\formone}{A}
\newcommand{\formtwo}{B}

\newcommand{\lneg}[1]{#1^{\bot}}
\newcommand{\dneg}[1]{#1^{\bot\bot}}
\newcommand{\contone}{\mathsf{C}}

\newcommand{\poscont}{\mathsf{P}}
\newcommand{\negcont}{\mathsf{N}}

\newcommand{\mll}{\mathsf{MLL}}
\newcommand{\stone}{\mathsf{S}}
\newcommand{\sttwo}{\mathsf{R}}
\newcommand{\stthree}{\mathsf{T}}
\newcommand{\stfour}{\mathsf{U}}

\newcommand{\judgone}{J}
\newcommand{\pmllone}{\xi}
\newcommand{\pmlltwo}{\mu}
\newcommand{\maptwo}[1]{\mathscr{I}(#1)}
\newcommand{\smllone}{\mathsf{J}}
\newcommand{\smlltwo}{\mathsf{H}}
\newcommand{\smllthree}{\mathsf{G}}

\newcommand{\bitocc}[1]{\mathcal{B}(#1)}
\newcommand{\bitval}[1]{\mathcal{V}(#1)}
\newcommand{\transmll}[1]{\mapsto_{#1}}

\newcommand{\mapthree}[3]{\mathscr{R}_{#1,#2}(#3)}

\newcommand{\mmllone}[1]{\mathcal{M}_{#1}}
\newcommand{\stgen}{\mathsf{S}}
\newcommand{\statesmll}[1]{\mathsf{T}_{#1}}

\newcommand{\pfun}[1]{[#1]}
\newcommand{\vecone}{x}

\newcommand{\midd}{\; \; \mbox{\Large{$\mid$}}\;\;}

\newenvironment{varitemize}
{
\begin{list}{\condinc{\labelitemii}{\labelitemi}}
{\setlength{\itemsep}{0pt}
 \setlength{\topsep}{0pt}
 \setlength{\parsep}{0pt}
 \setlength{\partopsep}{0pt}
 \setlength{\leftmargin}{15pt}
 \setlength{\rightmargin}{0pt}
 \setlength{\itemindent}{0pt}
 \setlength{\labelsep}{5pt}
 \setlength{\labelwidth}{10pt}
}}
{
 \end{list} 
}

\newcounter{numberone}

\newenvironment{varenumerate}
{
\begin{list}{\arabic{numberone}.}
{
  \usecounter{numberone}
  \setlength{\itemsep}{0pt}
  \setlength{\topsep}{0pt}
  \setlength{\parsep}{0pt}
  \setlength{\partopsep}{0pt}
  \setlength{\leftmargin}{15pt}
  \setlength{\rightmargin}{0pt}
  \setlength{\itemindent}{0pt}
  \setlength{\labelsep}{5pt}
  \setlength{\labelwidth}{15pt}
}}
{
\end{list} 
}

\newcounter{numbertwo}

\newcommand{\urule}[3]{%
  \prooftree #1 \justifies #2 \using #3 \endprooftree}
\newcommand{\brule}[4]{%
  \prooftree #1\ \ \ #2 \justifies #3 \using #4 \endprooftree}

\newcommand{\tens}{\otimes} 
\newcommand{\lpar}{\parr} 

\condinc{}{
\newtheorem{theorem}{Theorem}
\newtheorem{proposition}{Proposition}
\newtheorem{lemma}{Lemma}
\newtheorem{definition}{Definition}

\newtheorem{example}{Example}

\newtheorem{notation}{Notation}
\newenvironment{proof}{\begin{trivlist}
       \item[\hskip \labelsep {\bfseries Proof.}]}{\hfill $\Box$ \end{trivlist}}}

\newcommand{\VV}{\mathbb{V}}
\newcommand{\AAA}{\mathbb{A}}
\newcommand{\BB}{\mathbb{B}}
\newcommand{\NN}{\mathbb{N}}
\newcommand{\CC}{\mathbb{C}}

\input{Qcircuit}

\begin{document}

\condinc{
\title{Wave-Style Token Machines\\ and Quantum Lambda Calculi}

\author{Ugo Dal Lago\inst{1} \and Margherita Zorzi\inst{2}\thanks{Partially supported by LINTEL (Linear Techniques 
   For The Analysis Of Languages), \texttt{https://sites.google.com/site/tolintel}}}

\institute{Universit\`a di Bologna \& INRIA, \email{\texttt{dallago@cs.unibo.it}} 
\and Universit\`a degli Studi di Verona, Italy, \email{\texttt{margherita.zorzi@univr.it}}}
}{
\title{Wave-Style Token Machines and Quantum Lambda Calculi}
\author{Ugo Dal Lago\footnote{Universit\`a di Bologna \& INRIA} \and Margherita Zorzi\footnote{Universit\`a di Verona}}
\date{}
}

\maketitle

\begin{abstract}
Particle-style token machines are a way to interpret proofs and programs, when the latter are 
written following the principles of linear logic. In this paper, we show that token machines also make sense when the programs
at hand are those of a simple quantum $\lambda$-calculus. This, however,
requires generalizing the concept of a token machine to one in which more than one particle
travel around the term \emph{at the same time}. The presence of multiple tokens is intimately related to entanglement
and allows to give a simple operational semantics to the calculus, coherently with the principles
of quantum computation.
\end{abstract}

\section{Introduction}
One of the strongest trends in computer science is the (relatively recent) interest in exploiting 
new computing paradigms which go beyond the usual, classical one. Among these paradigms, quantum computing 
plays an important role. In particular, the quantum paradigm is having a deep impact on the 
notion of a computationally (in)tractable problem. In this respect, two of the most surprising results are due to Peter Shor, 
who proved that prime factorization of integers and the discrete logarithm can be efficiently solved 
(i.e. in polynomial time) by a quantum computer~\cite{Shor97}. 

Even if quantum computing has catalyzed the interest of a quite large scientific community, several theoretical 
aspects are still unexplored. As an example, the definition of a robust theoretical framework for quantum 
programming is nowadays still a challenge. A number of (paradigmatic) calculi for quantum computing have 
been introduced in the last ten years. Among them, some functional calculi, typed and untyped, 
have been proposed~\cite{mscs2009,tcs2010,entcs11,SV06,vT04}, but we are still at a stage where it is not clear whether one calculus
could be considered \emph{canonical}. Moreover, the meta-theory of most of these formalisms lack
the simplicity of the one of their ``classical'' siblings.

It is clear that linear logic and quantum computing are strongly related: since quantum data have to undergo 
restrictions such as no-cloning and no-erasing, it is not surprising that in most of the cited quantum 
calculi the use of resources is controlled. Linear logic therefore provides an ideal framework where rooting 
quantum data treatment, but also offers another tool which has not been widely exploited in the quantum setting: 
its mathematical model in terms of operator algebras, i.e. the Geometry of Interaction (GoI in the following).
Indeed, the latter provides a dynamical interpretation and a semantic account of the cut-elimination procedure 
as a flow of information circulating into a net structure. This idea can be formulated both as an algebra of 
bounded operators on a infinitely dimensional Hilbert space~\cite{JYG89} or as a token-based machine (a rewriting 
automata model with local transition rules)~\cite{GA92,Mackie95}. Both formulations seem to be promising in the quantum 
setting. On the one hand, the Hilbert space on top of which the first formulation of GoI is given is precisely the
canonical state space of a quantum Turing machine (see for example~\cite{BerVa97}). On the other hand, the 
definition of a token machine provides a mathematically simpler setting, which has already found a 
role in this context~\cite{DLF11,HH11}.

In this paper, we show that token machines are also a model of a linear quantum $\lambda$-calculus $\QL$ defined 
along the lines of van Tonder's $\lambda_q$~\cite{vT04}. This allows to give an operational semantics to $\QL$ which renders the quantum nature
of $\QL$ explicit: type derivations become quantum circuits built on exactly the set of gates occurring in the
underlying $\lambda$-term. This frees us from the burden of having to define the operational semantics of
quantum calculi in reduction style, which is known to be technically challenging in a similar setting~\cite{vT04}.
On the other hand, the power of $\beta$-style axioms is retained in the form of an equational theory for which
our operational semantics can be proved sound.

Technically, the design of our token machine for $\QL$, called $\IAM{\QL}$ is arguably more challenging than
the one of classical token machines. Indeed, the principles of quantum computing, and the so-called
\emph{entanglement} in particular, force us to go towards \emph{wave-style} machines, i.e., to machines where more
than one particle can travel inside the program at the same time. Moreover, the possibly many tokens at hand
are subject to synchronization points, each one corresponding to unitary operators
of arity greater than $1$. This means that $\IAM{\QL}$, in principle, could suffer from deadlocks, let alone
the possibility of non-termination. We here prove that these pathological situations can not happen. 

In Section~\ref{LLTM}, we recall the token machine for multiplicative linear logic. 
\condinc{}{In Section~\ref{sec:qcomp} we propose a gentle introduction to quantum computing.} The calculus $\QL$ and its token machine $\IAM{\QL}$ are 
introduced in Section~\ref{sec:QL} and Section~\ref{sec:QLTM}, respectively.  Main results about $\IAM{\QL}$ are in Section~\ref{sect:mainres}.
\condinc{An extended version of this paper with more details and a gentle introduction to quantum computing is available~\cite{EV}.}{}
\condinc{}
{Sections~\ref{sec:relw} and~\ref{sect:conclusions} are respectively devoted to related works and conclusion/future plans}.

\section{Linear Logic and Token Machines}\label{LLTM}
In this section, we give some ideas about the simplest token machine, namely the one for the propositional,
multiplicative fragment of linear logic. This not only encourages the unfamiliar reader to understand the basic concepts underlying
this concrete approach to the geometry of interaction, but will also be useful in the following, when proving
basic results about quantum token machines. More details can be found in~\cite{DR99,GA92}.

Let $\AAA=\{\atomone, \atomtwo,\ldots\}$ be a countable set of \emph{propositional atoms}. 
\emph{Formulas} of Multiplicative Linear Logic (\MLL) are given by the following grammar:
$$
\formone,\formtwo::= \atomone \;|\; \lneg{\atomone}\;|\; \formone\otimes\formtwo\;|\;\formone\lpar\formtwo.
$$
Linear negation can be extended to all formulas in the usual \condinc{De Morgan's style. }{way: 
\begin{align*}
  \lneg{(\lneg{\alpha})}&=\alpha;\\
  \lneg{\formone\otimes\formtwo}&=\lneg{\formone}\lpar\lneg{\formtwo};\\
  \lneg{\formone\lpar\formtwo}&=\lneg{\formone}\otimes\lneg{\formtwo}.
\end{align*}
This way, $\dneg{\formone}$ is just $\formone$.}
The one-sided sequent calculus for \MLL\ is very simple:
$$
  \urule{}{\vdash{\formone,\lneg{\formone}}}{\mathsf{ax}}
  \qquad
  \brule{\vdash{\conone,\formone}}{\vdash\contwo,\lneg{\formone}}{\vdash{\conone,\contwo}}{\mathsf{cut}}
  \qquad
  \brule{\vdash{\conone,\formone}}{\vdash\contwo,\formtwo}{\vdash{\conone,\contwo,\formone\otimes\formtwo}}{\otimes}
  \qquad
  \urule{\vdash{\conone,\formone,\formtwo}}{\vdash{\conone,\formone\lpar\formtwo}}{\lpar}
$$
The logic \MLL\ enjoys cut-elimination: there is a terminating algorithm turning any \MLL\ proofs into
a cut-free proof of the same conclusion.

Consider the following \MLL\ proof $\pmllone$ (where different occurrences of the same
propositional (co)atom have been numbered):
$$
\urule
  {
    \brule
      {
        \brule
        {\urule{}{\vdash\lneg{\atomone_4},\atomone_4}{\mathsf{ax}}}
        {\urule{}{\vdash\lneg{\atomone_5},\atomone_5}{\mathsf{ax}}}
        {\vdash\lneg{\atomone_3},\atomone_3}{\mathsf{cut}}
      }
      {
        \urule{}{\vdash\lneg{\atomtwo_3},\atomtwo_3}{\mathsf{ax}}
      }
      {\vdash\lneg{\atomone_2},\atomtwo_2,\atomone_2\otimes\lneg{\atomtwo_2}}{\otimes}
  }
  {\vdash\lneg{\atomone_1}\lpar\atomtwo_1,\atomone_1\otimes\lneg{\atomtwo_1}}{\lpar}
$$
The token machine for $\pmllone$ is a simple automaton whose internal state is nothing more than
an occurrence of a propositional (co)atom in $\pmllone$. This state evolves by ``following'' this
occurrence, keeping in mind that atoms go down, while coatoms go up. A run of the token
machine of $\pmllone$ is, as an example, the following one:
$$
\lneg{\atomone_1}\transmll{\pmllone}
\lneg{\atomone_2}\transmll{\pmllone}
\lneg{\atomone_3}\transmll{\pmllone}
\lneg{\atomone_4}\transmll{\pmllone}
\atomone_4\transmll{\pmllone}
\lneg{\atomone_5}\transmll{\pmllone}
\atomone_5\transmll{\pmllone}
\atomone_3\transmll{\pmllone}
\atomone_2\transmll{\pmllone}
\atomone_1.
$$
This tells us that the occurrences $\lneg{\atomone_1}$ and $\atomone_1$ are somehow related. Similarly,
one could find a run relating $\atomtwo_1$ to $\lneg{\atomtwo_1}$. Remarkably, these correspondences 
survive cut-elimination.

\condinc{
In general, negative occurrences in the conclusion of any proof $\pmllone$ can
be put in correspondence with positive occurrences in a bijective way, since the
number of occurrences in $\pmllone$ is anyway finite, and the are not any cycles.
It is this correspondence which is taken as the semantics of $\pmllone$. 
In the following (Section~\ref{sec:mllql}), we will denote as $\mmllone{\pmllone}$ 
the token machine corresponding to the \MLL\ proof $\pmllone$.
}{
All this can be formalized through the notion of a \emph{context}, which is an \MLL\ formula with a hole: 
$$
\contone::= [\cdot]\midd\contone\otimes\formone\midd\formone\otimes\contone\midd\contone\lpar\formone\midd\formone\lpar\contone.
$$
$\contone[\formone]$ is the formula obtained by replacing the unique occurrence of $[\cdot]$ in $\contone$
with $\formone$. If $\formone=\contone[\atomone]$ ($\formone=\contone[\lneg{\alpha}]$, respectively), we say that 
$\contone$ is a \emph{positive} (\emph{negative}, respectively) \emph{context for} $\formone$. 
If $\contone$ is positive (negative, respectively) for $\formone$, we sometime write
it as $\poscont_\formone$ (as $\negcont_\formone$, respectively).
An \emph{atom occurrence} in an \MLL\ proof $\pmllone$ is a pair $(\formone,\contone)$ 
where $\formone$ is an occurrence of an \MLL\ formula in $\pmllone$ and $\contone$ is a context for it.
Linear negation can be easily extended to contexts:
  \begin{align*}
    \lneg{[\cdot]}&=[\cdot];\\
    \lneg{(\contone\otimes \formtwo)}&=\lneg{\contone}\lpar\lneg{\formtwo}; &
    \lneg{(\formone\otimes\contone)}&=\lneg{\formone}\lpar\lneg{\contone};\\
    \lneg{(\contone\lpar \formtwo)}&=\lneg{\contone}\otimes\lneg{\formtwo}; &
    \lneg{(\formone\lpar\contone)}&=\lneg{\formone}\otimes\lneg{\contone}.
  \end{align*}
Please observe that $\contone$ is a negative context for $\formone$ iff
$\lneg{\contone}$ is a positive context for $\lneg{\formone}$. To every proof 
$\pmllone$ in $\mll$, we associate an automaton $\mmllone{\pmllone}$ which consists of:
\begin{varitemize}
\item 
  The finite set $\states{\pmllone}$ of \emph{states} of $\mmllone{\pmllone}$, which are all the atom occurrences of $\pmllone$;
\item 
  a \emph{transition relation} $\transmll{\pmllone}\subseteq\states{\pmllone}\times\states{\pmllone}$, \condinc{which can be easily
  defined along the lines described above.}{
  which is described by the rules in Figure~\ref{fig:trmll}. 
  \begin{figure}
    \fbox{
      \begin{minipage}{.97\textwidth}
     \begin{center}
       \begin{tabular}{cc}
         \begin{minipage}{1.5cm}
         \urule{}{\vdash{\formone,\lneg{\formone}}}{\mathsf{ax}}
         \end{minipage}
         &
         \begin{minipage}{4cm}
           $$
           \begin{array}{c}
             (\formone,\negcont_\formone)\transmll{\pmllone}(\lneg{\formone},\lneg{\negcont_\formone})\\
             (\lneg{\formone},\negcont_{\lneg{\formone}})\transmll{\pmllone}(\formone,\lneg{(\negcont_{\lneg{\formone}})})\\
           \end{array}
           $$
         \end{minipage}
       \end{tabular}

       \vspace{10pt}

       \begin{tabular}{cc}
         \begin{minipage}{3cm}
           \brule{\vdash{\conone_1,\formone}}{\vdash\contwo_1,\formtwo}{\vdash{\conone_2,\contwo_2,\formone\otimes\formtwo}}{\otimes}
         \end{minipage}
         &
         \begin{minipage}{4.5cm}
           $$
           \begin{array}{c}
             (\formone\otimes\formtwo, \negcont_\formone\otimes\formtwo)\transmll{\pmllone}(\formone,\negcont_\formone)\\
             (\formone\otimes\formtwo, \formone\otimes\negcont_\formtwo)\transmll{\pmllone}(\formtwo,\negcont_\formtwo)\\
             (\formone,\poscont_\formone)\transmll{\pmllone}(\formone\otimes\formtwo, \poscont_\formone\otimes\formtwo)\\
             (\formtwo,\poscont_\formtwo)\transmll{\pmllone}(\formone\otimes\formtwo, \formone\otimes\poscont_\formtwo)\\
             (\Gamma_2,\negcont)\transmll{\pmllone}(\Gamma_1,\negcont)\\
             (\Delta_2,\negcont)\transmll{\pmllone}(\Delta_1,\negcont)\\
             (\Gamma_1,\poscont)\transmll{\pmllone}(\Gamma_2,\poscont)\\
             (\Delta_1,\poscont)\transmll{\pmllone}(\Delta_2,\poscont)\\
           \end{array}
           $$
         \end{minipage}
       \end{tabular}

       \vspace{10pt}

       \begin{tabular}{cc}
         \begin{minipage}{2.5cm}
           \urule{\vdash{\conone_1,\formone,\formtwo}}{\vdash{\conone_2,\formone\lpar\formtwo}}{\lpar}
         \end{minipage}
         &
         \begin{minipage}{4cm}
           $$
           \begin{array}{c}
             (\formone\lpar\formtwo, \negcont_\formone\lpar\formtwo)\transmll{\pmllone}(\formone,\negcont_\formone)\\
             (\formone\lpar\formtwo, \formone\lpar\negcont_\formtwo)\transmll{\pmllone}(\formtwo,\negcont_\formtwo)\\
             (\formone,\poscont_\formone)\transmll{\pmllone}(\formone\lpar\formtwo, \poscont_\formone\lpar\formtwo)\\
             (\formtwo,\poscont_\formtwo)\transmll{\pmllone}(\formone\lpar\formtwo, \formone\lpar\poscont_\formtwo)\\
             (\Gamma_2,\negcont)\transmll{\pmllone}(\Gamma_1,\negcont)\\
             (\Gamma_1,\poscont)\transmll{\pmllone}(\Gamma_2,\poscont)\\
           \end{array}
           $$
         \end{minipage}
       \end{tabular}

       \vspace{10pt}

       \begin{tabular}{cc}
         \begin{minipage}{3.5cm}
           \brule{\vdash{\conone_1,\formone}}{\vdash\contwo_1,\lneg{\formone}}{\vdash{\conone_2,\contwo_2}}{\mathsf{cut}}
         \end{minipage}
         &
         \begin{minipage}{4cm}
           $$
           \begin{array}{c}
             (\formone,\poscont_\formone)\transmll{\pmllone}(\lneg{\formone},\lneg{(\poscont_\formone)})\\
             (\lneg{\formone},\poscont_{\lneg{\formone}})\transmll{\pmllone}(\formone,\lneg{(\poscont_{\lneg{\formone}})})\\
             (\Gamma_2,\negcont)\transmll{\pmllone}(\Gamma_1,\negcont)\\
             (\Delta_2,\negcont)\transmll{\pmllone}(\Delta_1,\negcont)\\
             (\Gamma_1,\poscont)\transmll{\pmllone}(\Gamma_2,\poscont)\\
             (\Delta_1,\poscont)\transmll{\pmllone}(\Delta_2,\poscont)\\
           \end{array}
           $$
         \end{minipage}
       \end{tabular}
       \vspace{10pt}
    \end{center}
    \end{minipage}}
\caption{Defining Rules for $\transmll{\pmllone}$}\label{fig:trmll}
\end{figure}}
\end{varitemize}
An atom occurrence in $\pmllone$ is said to be \emph{initial} (respectively, \emph{final}) iff
it is in the form $(\formone,\negcont_\formone)$ (respectively, in the form
$(\formone,\poscont_\formone)$), where $\formone$ is one among the formulas among
the conclusions of $\pmllone$. It is easy to verify that:
\begin{varitemize}
\item
  for every non-final occurrence $\occone$ there is exactly one occurrence $\occtwo$ such
  that $\occone\transmll{\pmllone}\occtwo$;
\item
  for every non-initial occurrence $\occone$ there is exactly one occurrence $\occtwo$ such
  that $\occtwo\transmll{\pmllone}\occone$.
\end{varitemize}
As a consequence, 
every initial occurrence is put in correspondence with a final occurrence in a bijective way --- the
number of occurrences in $\pmllone$ is anyway finite, and cycles cannot be reached from initial occurrences.
It is this correspondence which is taken as the semantics of $\pmllone$, after being shown to 
be invariant by cut-elimination. 

One last observation is now in order. Suppose $\occone_1,\ldots,\occone_n$ are \emph{all} the initial 
occurrences for $\pmllone$. Then, every occurrence in $\pmllone$ is visited \emph{exactly once} along one 
of the $n$ maximal computations starting in $\occone_1,\ldots,\occone_n$. This can be proved as follows:
\begin{varitemize}
\item
  First, prove the statement for any cut-free proof $\pmllone$, by induction on the
  structure of $\pmllone$;
\item
  Then show that if $\pmllone$ has the property and $\pmlltwo$ reduces to $\pmllone$
  by cut-elimination, $\pmlltwo$ has the property, too.
\end{varitemize}
Incidentally, this shows that cylic $\transmll{\pmllone}$ is acyclic.
}

\condinc{}
{
{
\section{Quantum Computing in a Nutshell}\label{sec:qcomp}
Quantum computing principles are  non-standard notions to the largest part of the ``lambda community''. 
The aim of this section is to provide to the non-expert reader an overview of quantum computing basic concepts. 
This will guide her or him in understanding the ``quantum content'' of our calculus (in particular, the meaning of 
unitary steps and the linear management of quantum data, see Section~\ref{sec:QL}). Moreover, notions like quantum 
entanglement, a peculiar feature of quantum data, offers some intuitions about how and why the choice of a wave-style 
token machine as operational model is the right choice.

The simplest quantum system is a  two-dimensional state space whose elements are called \emph{quantum bits} or 
\emph{qubits} for short. The qubit is the most basic unit of quantum information. The most direct way to represent 
a quantum bit is as a unitary vector
in the 2-dimensional Hilbert space $\ell^{2}(\{0,1\})$, which is isomorphic to $\CC^{2}$.  
We will denote with
$\ket{0}$ and $\ket{1}$  the elements of the computational basis of $\ell^{2}(\{0,1\})$.
The states $\ket{0}$ and $\ket{1}$ of a qubit correspond to the boolean constants $0$ and $1$, which
are the only possible values of a classical bit. A qubit, however, can assume other values, 
different from $\ket{0}$ and $\ket{1}$.  In fact, every linear combination
$\ket{\psi}=\alpha \ket{0} + \beta \ket{1}$
where $\alpha,\beta\in{\CC}$, and $|\alpha|^{2}+|\beta|^{2}=1$, represents
a possible qubit state. These states are
said to be \emph{superposed}, and the two values $\alpha$ and $\beta$ are called
\index{amplitudes}\emph{{amplitudes}}.
The amplitudes $\alpha$ and $\beta$  univocally represent the qubit with respect to the computational basis.
Given a qubit $\ket{\psi}=\alpha \ket{0} + \beta \ket{1}$, we commonly denote it by the vectorial notation
$$
\psi= 
\left( 
\begin{array}{cc}
\alpha  \\
\beta \\
 \end{array} 
 \right).
 $$
In particular, the vectorial representation of the elements of the computational basis $\ket{0}$ and $\ket{1}$ is the following:
$$
\left( 
\begin{array}{cc}
1  \\
0 \\
 \end{array} 
 \right)
 \qquad
 \left( 
\begin{array}{cc}
0 \\
1 \\
 \end{array} 
 \right)
 $$
While we can determine the state of a classical bit, for a qubit we
can not establish with the same precision the values $\alpha$ and $\beta$:
quantum mechanics says that a measurement of a qubit with state $\alpha \ket{0}+ \beta
\ket{1}$ has the effect of changing the state to $\ket{0}$ with probability
$|\alpha|^{2}$ and to $\ket{1}$ with probability $|\beta|^{2}$. 
For example, if $\ket{\psi}=\frac{1}{\sqrt{2}}\ket{0}+\frac{1}{\sqrt{2}}\ket{1}$, one can 
observe $0$ or $1$ with the same probability $|\frac{1}{\sqrt{2}}|^{2}=\frac{1}{2}$. In this 
brief survey on quantum computing, we will not enter in the details about qubit measurement, 
since the syntax of the calculus $\QL$ does not include an explicit measurement operator 
(a constant whose --- probabilistic --- operational semantics mimics the observation of 
quantum data). This choice is sound from a theoretical viewpoint, since it is possible to 
assume to have a unique, final measurement, at the end of the computation. Notwithstanding, 
the measurement operator is a useful programming tool in order to encode quantum algorithms and 
the extension of the syntax with a measurement operator is one of our planned future works.
For a complete overview about measurement of qubits and relationships between different kind 
of measurement, see~\cite{NieCh00}.

In order to define arbitrary set of quantum data, we need a generalization
of the notion of qubit,  called  \index{quantum register}\textit{{quantum register}} or, more commonly, 
\index{quantum state}\emph{{quantum state}}~\cite{vT04,SV06,NiOz09}. A quantum register can be viewed 
as a system of $n$ qubits and, mathematically, it is a normalized vector in the Hilbert space 
$\ell^{2}(\{0,1\}^{n})$  ($\{0,1\}^{n}$ is a compact notation to represent any binary sequence of length $n$).
The standard computational basis for $\ell^{2}(\{0,1\}^{n})$ is $\mathcal{B}=\{ \ket{i}\ |\ i \mbox{ is a binary string of length } n \}$.
\begin{notation}
We use the notation $\ket{b_1\ldots b_k}$ ($b_i\in\{0,1\}$)  for $\ket{b_1}\otimes\ldots\otimes\ket{b_k}$, 
where $\otimes$ is the \emph{tensor product} (see below).
\end{notation}
With a little abuse of language, we say that the number of quits $n$ corresponds to the dimension of the space. Notice that if the dimension is $n$, then the basis $\mathcal{B}$ contains $2^{n}$ elements, and each quantum states is a normalized  linear combination of these elements:
$$\alpha_{1}{{\underbrace{\ket{00\ldots0}}_{n}}}+\alpha_{2}\ket{00\ldots1}+\ldots +\alpha_{2^{n}}\ket{11\ldots1}$$

\begin{example}
 Let us consider a 2-level quantum system, i.e. a system of two qubits.
Each 2-qubit quantum register is a normalized vector in $\ell^2(\{0,1\}^{2})$ and the computational basis is  $\{\ket{00}, \ket{01}, \ket{10}, \ket{11}\}$.
For example,  $\frac{1}{\sqrt{2}}\ket{00}+ \frac{1}{\sqrt{4}}\ket{01}+ \frac{1}{\sqrt{8}}\ket{10}+ \frac{1}{\sqrt{8}}\ket{11}$
is a quantum register of two qubits and we can represent it as
$$
\psi= 
\left( 
\begin{array}{c}
\frac{1}{\sqrt{2}}  \\
\frac{1}{\sqrt{4}}\\
\frac{1}{\sqrt{8}}\\
\frac{1}{\sqrt{8}}
 \end{array} 
 \right).
 $$
\end{example}

An Hilbert space of dimension $n$ can be built from smaller Hilbert spaces by means of the \emph{tensor product} $\otimes$. If $H_1$ is an Hilbert space of dimension $k$ and $H_2$ is  an Hilbert space of dimension $m$, $H_3=H_1\otimes H_2$ is an Hilbert space of dimension $km$ (each element is a vector of $km$ coordinates obtained by ``hooking'' a vector in $H_2$ to a vector in $H_1$). 
In other words, an $n$-qubit quantum register with $n\geq 2$ can be viewed as a composite system. 
It is possible to combine two (or more) distinct  physical systems into a composite one. 
If the first system is in the state $\ket{\phi_1}$ (a vector in a Hilbert Space $H_1$) and the second system is in the state $\ket{\phi_2}$ (a vector in a Hilbert Space $H_1$) , then the state of the combined system is $\ket{\phi_1}\otimes\ket{\phi_2}$ (a vector in a Hilbert Space $H_1\otimes H_2$) .\\
We will often omit the ``$\otimes$'' symbol, and will write the joint state  as $\ket{\psi_1} \ket{\psi_2}$ or as $\ket{\psi_1\psi_2}$.

Not all quantum states can be viewed as composite systems: this case occurs in presence of \emph{entanglement} phenomena (see below).
Since normalized vectors of quantum data represent physical systems, the (discrete) evolution of systems can be viewed as a suitable transformation on Hilbert spaces. The evolution of a quantum register  is \emph{linear} and \emph{unitary}.
Giving an initial state $\ket{\psi_1}$, for each evolution to a state $\ket{\psi_2}$, there exists a \emph{unitary}  operator $U$ such that $\ket{\psi_2} = U \ket{\psi_1}$. 
Informally, ``unitary'' referred to an algebraic operator on a suitable space means that the normalization constraint of the amplitudes ($\sum_i|\alpha_i|^{2}=1$) is preserved during the transformation. 
Thus, a quantum physical system, i.e. a normalized vector which represents our data, can be described in term of linear operators and in a \emph{deterministic} way. 
In quantum computing we refer to a unitary operator $U$ acting on a $n$-qubit quantum register as an $n$-qubit \index{quantum gate}\emph{{quantum gate}}. 
We can represent operators on the $2^{n}$-dimensional Hilbert space  $\ell^{2}(\{0,1\}^{n})$ with respect to the standard basis of $\CC^{2^{n}}$ as  $2^{n}\times 2^{n}$ matrices, and it is possible to prove that to each unitary operator on a Hilbert Space it is possible to associate an algebraic representation.
Matrices which represent unitary operators enjoy some important property: for example they are easily invertible (reversibility is one of the peculiar features of quantum computing). 
\emph{The application  of quantum gates to quantum registers  represents the pure quantum computational step and captures the internal evolution of quantum systems.}
The simplest quantum gates act on a single qubit: they are operators on the space $\ell^2(\{0,1\}^{})$, represented in $\CC^{2}$ by $2\times 2$ complex matrices.
For example, the quantum gate $X$ is the unitary operator which maps $\ket{0}$ to $\ket{1}$ and  $\ket{1}$ to $\ket{0}$ and it is represented by the matrix
$$\left( 
\begin{array}{cc}
0 & 1  \\
1 & 0 \\
 \end{array} 
 \right)$$
Being a linear operator, it maps a linear combination of inputs to the corresponding linear combination of outputs, and so $X$ maps the general qubit state 
$\alpha \ket{0}+\beta \ket{1}$ into the state  $\alpha \ket{1}+\beta \ket{0}$ i.e 
$$
\left( 
\begin{array}{cc}
0 & 1  \\
1 & 0 \\
 \end{array} 
 \right)
 \left( 
\begin{array}{c}
\alpha  \\
\beta \\
 \end{array} 
 \right)
 =
  \left( 
\begin{array}{c}
\beta  \\
\alpha \\
 \end{array} 
 \right)
$$
An interesting unitary gate is the \index{Hadamard gate}\emph{{Hadamard gate}} denoted by H
which acts on the computational basis in the following way: 
$$\ket{0}\mapsto \frac{1}{\sqrt{2}}(\ket{0}+\ket{1})\qquad  \ket{1}\mapsto \frac{1}{\sqrt{2}}(\ket{0}-\ket{1})$$\\
The Hadamard gate, which therefore is given by the matrix
$$
H= \frac{1}{\sqrt{2}}\left( 
\begin{array}{cc}
1 & 1   \\
1 & -1 \\
 \end{array} 
 \right)
 $$
is  useful when we want to create a superposition starting from a classical state.  It also holds that $H(H(\ket{c}))=\ket{c}$ for $c=\{0,1\}$. 
1-qubit quantum gates can be used in order to build gates acting on $n$-qubit quantum states.  
If we have a 2-qubit quantum system, we can apply a 1-qubit quantum gate only to one component of the system, and we implicitly apply 
the identity operator (the identity matrix) to the other one. For example suppose we want to apply X  to the first qubit. The 2-qubits 
input $\ket{\psi_1}\otimes\ket{\psi_2}$ gets mapped to $X\ket{\psi_1}\otimes I \ket{\psi_2}=(X\otimes I)\ket{\psi_1}\otimes\ket{\psi_2}$. \\

The \textbf{CNOT}  is one of the most important quantum operators.
  It is mathematically described by the standard operator
 ${CNOT}:\ell^2({\{0,1\}^{2}})\rightarrow{\ell^2({\{0,1\}^{2}})}$ defined by 
 \vspace{-.2cm}
\begin{center}
\begin{minipage}{3cm}
\begin{eqnarray*}
  \mathbf{CNOT}\ket{0 0}&=&\ket{0 0}\\
  \mathbf{CNOT}\ket{0 1}&=&\ket{0 1}\\
\end{eqnarray*}
\end{minipage}
\hspace{1cm}
\begin{minipage}{3cm}
\begin{eqnarray*}
  \mathbf{CNOT}\ket{1 0}&=&\ket{1 1}\\
  \mathbf{CNOT}\ket{1 1}&=&\ket{1 0}\\
\end{eqnarray*}
\end{minipage}
\end{center}
\vspace{-.2cm}
Intuitively, $\mathbf{cnot}$ acts as follows: it takes two distinct quantum bits as inputs and complements the \emph{target}
  bit (the second one) if the \emph{control} bit (the first one)  is 1; otherwise it does not
  perform any action.
\comment{  The matrix which describes it is the following:
$$ \left( 
 \begin{array}{cccc}
1 &  0 & 0 & 0 \\
0 & 1 & 0 & 0 \\
0 & 0  & 0 & 1 \\
0 &  0 & 1 & 0 \\
 \end{array} 
 \right)
$$
}
The control qubit is a ``master'' agent:  its evolution in independent from the evolution of the target bit (if the first input of the cnot is $\ket{\phi}$ the output is the same); the target qubit is a ``slave'' agent: its evolution is controlled by the value of the first qubit.  In some sense, a communication between the agents is required and the quantum circuit is a simple distributed system.
By adopting this perspective,  controlled operators like cnot acts as ``synchronization points'' between token (ground type occurrences) in our definition of quantum token machine: this is one of the main features of our semantics (see Section~\ref{sec:QLTM}).

\bigskip
\comment{
In a realistic perspective, a quantum system interacts with other systems and also with a measurement apparatus. This means that we can ``observe'' the amplitudes $\alpha_i$ obtaining a quantitative information about the state.
In the following, we will describe two kind of total and partial measurements of a quantum state through some examples.
Let $B=\{\ket{\phi_i}\}$ be an  basis of a Hilbert state space ${H}$ (remember that $\phi_i$ is a finite  binary sequence). It is possible to perform a \emph{measurement} on ${H}$ w.r.t. $B$ such that, given a state $\ket{\psi}=\sum_{i}\alpha_i \ket{\phi_i}$, the measurement leaves the system in the state $\phi_i$ with probability $|\alpha_i|^{2}$. 
\begin{example}
Assume $\ket{\psi}=\frac{1}{\sqrt{2}}\ket{00}+\frac{1}{\sqrt{4}}\ket{01}+\frac{1}{\sqrt{8}}\ket{10}+\frac{1}{\sqrt{8}}\ket{11}$. A total measurement implies the following situations: one can observe $\ket{00}$ with probability $|\frac{1}{\sqrt{2}}|^{2}=\frac{1}{2}$ or $\ket{01}$ with probability $|\frac{1}{\sqrt{4}}|^{2}=\frac{1}{4}$ or $\ket{10}$ with probability $|\frac{1}{\sqrt{8}}|^{2}=\frac{1}{8}$ or $\ket{11}$ with probability $|\frac{1}{\sqrt{8}}|^{2}=\frac{1}{8}$.
\end{example}
The described measurement 
is a special kind of total  \emph{projective} measurements (see e.g. \cite{NieCh00}).
Total projective measurement (here proposed in a simplified setting) is intuitive and it is commonly used. 
It enjoys some properties. For example, repeated measurements give the same result: if the after a measurement the system is (with the related probability) in the state $\phi_i$, a second measurement applied on  $\phi_i$ acts as the identity (reason on the previous example). Moreover, total projective measurement destroys the superposition: the state $\ket{\psi}$ collapses to an element of the  basis, totally loosing the quantum superposition (but notice that it does not modify the dimension of the  Hilbert space). It is possible, however, to define projective measurement also with respect to a single qubit, i.e. it is possible to observe the i-th quit of the basis vectors. In this case, only a part of the superposition is destroyed, possibly leaving the system in quantum superposition. See the following subsection.

It is also possible to perform a partial measurement: this means that we observe only a subset of the qubits and, as a consequence, we destroy only a part of the quantum superposition.  
\begin{example}
Assume again $\ket{\psi}=\frac{1}{\sqrt{2}}\ket{00}+\frac{1}{\sqrt{4}}\ket{01}+\frac{1}{\sqrt{8}}\ket{10}+\frac{1}{\sqrt{8}}\ket{11}$. Let us focus our attention on the first qubits and let us measure it.
One can observe $0$ with probability $|\frac{1}{\sqrt{2}}|^{2}+|\frac{1}{\sqrt{2}}|^{2}=\frac{1}{2}+\frac{1}{4}=\frac{3}{4}$ or $1$ with probability $|\frac{1}{\sqrt{8}}|^{2}+|\frac{1}{\sqrt{8}}|^{2}=\frac{1}{8}+\frac{1}{8}=\frac{1}{4}$: in the first case, the partial measurement leaves the original state $\ket{\psi}$ in the state $\ket{\psi_1}= k\cdot(\frac{1}{\sqrt{2}}\ket{00}+\frac{1}{\sqrt{4}}\ket{01})$ (where $k$ is a normalizing factor, since the squared moduli of the amplitudes have to sum up to 1); in the second case,  the partial measurement leaves the original state $\ket{\psi}$ in the state $\ket{\psi_2}=h\cdot(\frac{1}{\sqrt{8}}\ket{10}+\frac{1}{\sqrt{8}}\ket{11})$ (where $h$ is again a normalizing factor). 
\end{example}
We notice  that also in this case the measurement does not modify the dimension of the Hilbert space.
Notice also that after the measurement we can perform a new measurement of the first qubit which acts as the identity (we can only observe $0$ with probability $1$); otherwise, one can  measure the second qubit: in this case the measurement is significative (one can  observe both $0$ and $1$).
}

Not all quantum states can be viewed as composite systems. In other words, \textit{if $\ket{\psi}$ is a state of a tensor product space $\mathcal{H}_1\otimes\mathcal{H}_2$, it is not generally true that there exists $\ket{\psi_1}\in\mathcal{H}_1$ and $\ket{\psi_2}\in\mathcal{H}_2$ such that $\ket{\psi}=\ket{\psi_1}\otimes \ket{\psi_2}$}. Instead, \textit{ it is not always possible to decompose an $n$-qubit register as the tensorial product of 
$n$ qubits}.\\
These non-decomposable registers are called {\itshape entangled} and enjoy 
properties that we cannot find in any object of classical physics (and therefore in classical data).
If  $n$ qubits are entangled, they behave \emph{as if connected},
independently of the real physical distance.
The strength of quantum computation is essentially based on the 
existence of entangled states (see, for example, the \emph{teleportation protocol} \cite{NieCh00}).
\begin{example}\label{ex:cnotex}
The 2-qubit states $\ket{\psi}=\frac{1}{\sqrt{2}} \ket{00}+ \frac{1}{\sqrt{2}}\ket{11}$ and $\ket{\psi}=\frac{1}{\sqrt{2}} \ket{01}+ \frac{1}{\sqrt{2}}\ket{10}$ are entangled. The 2-qubit state $\ket{\phi}=\alpha\ket{00}+\beta\ket{01}$ is not entangled. Trivially, notice that it is possible to rewrite it in the mathematically equivalent form $\phi=\ket{0}\otimes(\alpha\ket{0}+\beta\ket{1})$. 
\end{example}
A simple way to create an entangled state is to fed a $\mathbf{CNOT}$ gate with a target qubit $\ket{c}$ and a particular control qubit, more precisely the output of the Hadamard gate applied to a base qubit, therefore a superposition $\frac{1}{\sqrt{2}}\ket{0}+\frac{1}{\sqrt{2}}\ket{1}$ or $\frac{1}{\sqrt{2}}\ket{0}-\frac{1}{\sqrt{2}}\ket{1}$. This  composition of quantum gates is actually encoded by the terms defined in the Example~\ref{ex:epr}.

We previously said that each $n$-ary unitary transformation (or composition of unitary transformations) can be represented by a suitable $n\times n$ matrix. From a computer science viewpoint, it is common to reason about quantum states transformations in terms of \emph{quantum circuits}. Through the paper, we frequently say that ``a lambda term encodes a quantum circuit''. What does this mean? What is a quantum circuit? One more time, this is a long and complex subject and we refer to~\cite{NieCh00,NaOh08} for a complete and exhaustive explanation.
Since quantum circuits are invoked in the proof of Soundness Theorem~\ref{th:sound}, we give here some intuitions and a qualitative description (enough to understand the Soundness proof) of quantum circuits.
We have introduced qubits to store quantum information, in analogy with the classical case. We have also introduced operations acting on them, i.e. quantum gates, and  we can think about quantum gates in analogy with gates in classical logic circuits. 


A quantum circuit on $n$ qubits  implements an unitary operator on a Hilbert space of dimension $\CC^{2^{n}}$. This can be views as a primitive collection of quantum gates, each implementing a unitary operator on k (small) qubits.

It is useful to graphically represent quantum circuit in terms of sequential and parallel composition of quantum gates and wires, as for boolean circuits (notwithstanding, in the quantum case the graphical representation does not reflect the physical realization of the circuit).

For example, the following diagram represents the quantum circuit implemented by the term in Example~\ref{ex:epr}.

\vspace{-5ex}
{\scriptsize
     \begin{center}
       \scalebox{0.25}
       {
         \ifx\pdfoutput\undefined 
         \epsfbox{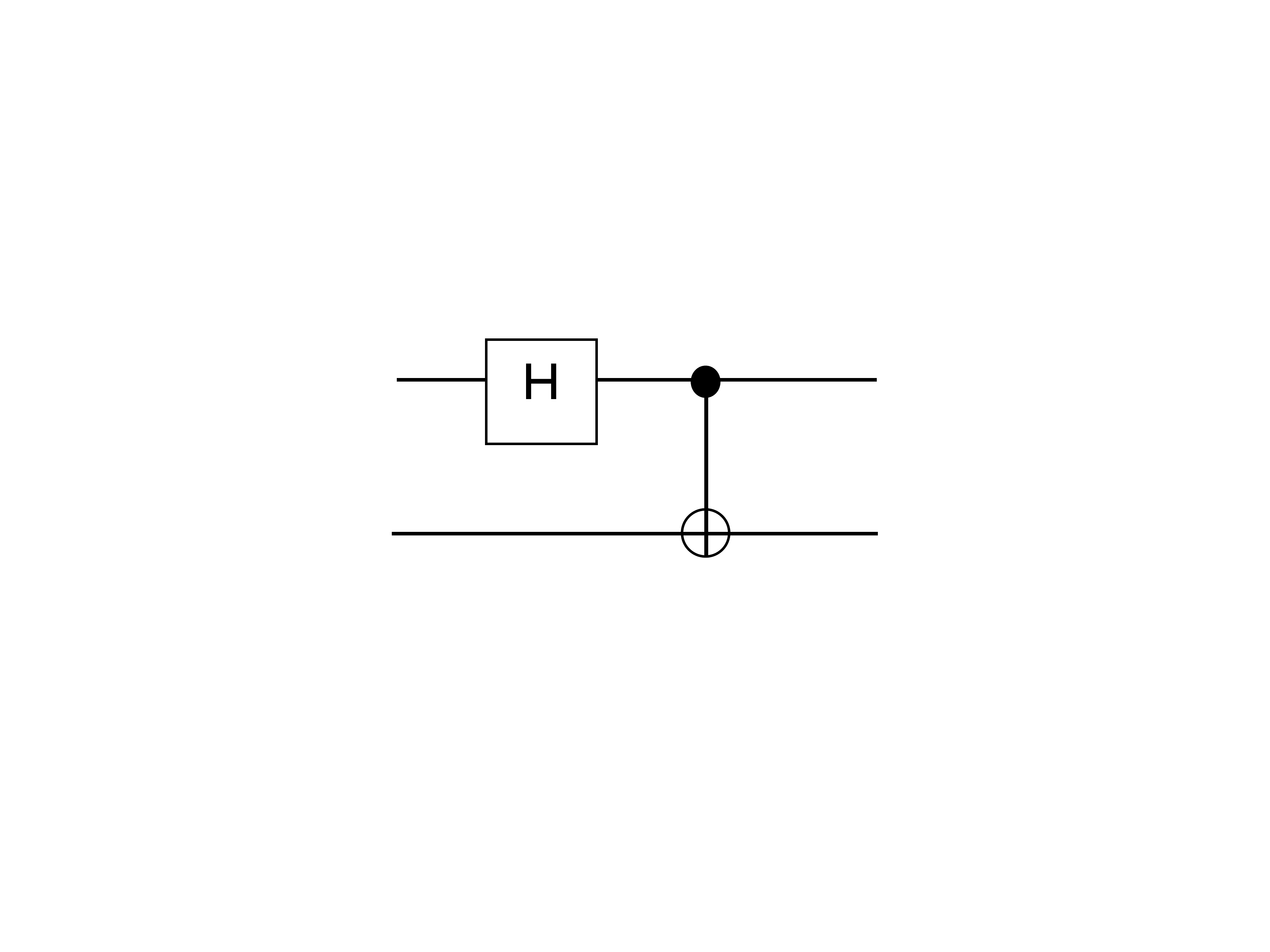} \else 
         \includegraphics{epr} \fi
       }
     \end{center}
     }

The calculus $\QL$ is purely linear (see Section~\ref{sec:QL}). Each (well typed) lambda terms encode a quantum transformation or, equivalently, a quantum circuit built  on the set of (the constants representing) quantum gates occurring in the lambda-term.

\bigskip
One of the primitive operations in information theory is the copy of a datum. When we deal with quantum data as qubits, quantum information suffers from lack of accessibility in comparison to classical one.
In fact, a quantum bit \emph{can not be duplicated}. This curious feature is well-know in literature as \emph{no-cloning property}:  it does not allow to make a copy of an unknown  quantum state (it is only possible to duplicate ``trivial'' qubits, i.e. basis states $\ket{0}$ and $\ket{1}$). In other words, it is not possible to build a quantum transformation/a quantum circuit able to maps an arbitrary quantum state $\ket{\psi}$ into the state $\ket{\psi}\otimes\ket{\psi}$.
No-cloning property is one of the main difference between classical and quantum data and any paradigmatic quantum language has deal with to this fact. Notwithstanding, even if no-cloning property made the design of quantum languages more challenging, quantum data enjoy some properties (which have no classical counterpart) which can be exploited in the design of quantum algorithms.

\comment{
\begin{example}[Entanglement and measurement]
The effects of entanglement phenomena can be easily observed when we perform a (partial) measurement. Let us consider the entangle state $\ket{\psi}=\frac{1}{\sqrt{2}} \ket{00}+ 1/\sqrt{2}\ket{11}$. 
If we measure the first qubit obtaining, for example, $0$ (with the related probability $\frac{1}{{2}}$), we loose the possibility to obtain $1$ by measure of the second qubit: in fact, the state collapses to  $\ket{\psi'}=k\cdot(\frac{1}{\sqrt{2}} \ket{00})$ and the only possible result of a measurement of the second qubit is actually $0$.
In some sense, the measurement of the second qubit has been influenced by the measurement
 of the first one. 
\end{example}
}
}
}
\condinc{}
{
}
\section{The Calculus $\QL$}\label{sec:QL}
An essential  property of quantum programs is that quantum data, i.e. quantum bits, should always 
be uniquely referenced. This restriction follows from the well-known \emph{no-cloning} and \emph{no-erasing} properties 
of quantum physics, which state that a quantum bit cannot be duplicated nor canceled~\cite{NieCh00}. 
Syntactically, one captures this restriction by means of linearity: 
if every abstraction $\abstr{x}{\termone}$ is such that there is exactly one free occurrence of 
$x$ in $\termone$, then the substitution triggered by firing \emph{any} redex is neither copying nor cancelling
and, as a consequence, coherent with the just stated principles.

In this Section, we introduce a quantum linear $\lambda$-calculus in the style of van Tonder's $\lambda_q$~\cite{vT04} and give an equational theory for
it. This is the main object of study of this paper, and is the calculus for which we will give a wave-style token
machine in the coming sections.
\subsection{The Language of Terms}\label{sec:syntaxql}
Let us fix a finite set $\uset$ of \emph{unitary operators}, each on a finite-dimensional Hilbert space 
$\CC^{2^n}$, where $n$ can be arbitrary. To each such $\unopone\in\uset$ we associate 
a symbol $\opone$ and call $n$ the \emph{arity} of $\opone$. The syntactic categories of 
\emph{patterns}, \emph{bits}, \emph{constants} and \emph{terms} are defined by the 
following grammar:
$$
\begin{array}{lcl}
  \patone & ::= & \varone\ \mid\ \pair{\varone}{\vartwo};  \hfill \mbox{\emph{patterns}}\\ 
  \cbitone & ::= & \ket{0}_n \mid\ \ket{1}_n;   \hfill \mbox{\emph{bits}}\\ 
  \csone & ::= & \cbitone\ \mid\ \opone; \hfill\mbox{\emph{constants}}\\ 
  \termone,\termtwo & ::= & \varone \mid \csone \mid \termone\otimes\termtwo 
     \mid \termone\termtwo\ \mid\ \lambda \patone.\termone;\qquad\qquad\qquad\hfill \mbox{\emph{terms}}
\end{array}
$$
where $n$ ranges over $\NN$ and $\varone$ ranges over a denumerable, totally ordered
set of variables $\VV$. We always assume that the natural numbers occurring
next to bits in any term $\termone$ are pairwise distinct. This condition, by the way,
is preserved by substitution when the substituted variable occurs (free) exactly once.
Whenever this does not cause ambiguity, we elide labels and simply write
$\ket{\bitone}$ for a bit. Notice that pairs can be formed via the binary operator $\otimes$.
We will sometime write $\ket{\bitone_1 \bitone_2\ldots \bitone_k}$ for 
$\ket{\bitone_1}\otimes\ket{\bitone_2}\otimes\ldots\otimes\ket{\bitone_k}$ 
(where $\bitone_1,\ldots,\bitone_n\in\{0,1\}$). In the following, capital 
letters such as $\termone$, $\termtwo$, $\termthree$, $\termfour$ (possibly 
indexed), denote terms. We work modulo variable renaming; in other words, 
terms are equivalence classes modulo $\alpha$-conversion. 
Substitution up to $\alpha$-equivalence is defined in the usual way. 
Observe that the terms of $\QL$ are the ones of a $\lambda$-calculus with pairs
(which are accessed by pattern-matching) endowed with constants for bits and unitary
operators. We don't consider measurements here, and discuss the possibility of
extending the language of terms in Section~\ref{sect:conclusions}.
\subsection{Judgements and Typing Rules}\label{sec:typing}
Since in $\QL$ all terms are assumed to be non-duplicable by default, we adopt 
a linear type-discipline. Formally, the set of types is defined as   
$$
\typeone::=\BB\mid\ \typeone\linmap\typetwo\mid\typeone\otimes\typetwo,
$$
where $\BB$ is the ground type of qubits.
We write $\BB^n$ for the $n$-fold tensor product 
\condinc{
$\BB\otimes\ldots\otimes\BB$ ($n$ times).
}{
$$
\underbrace{\BB\otimes\ldots\otimes\BB}_{\mbox{$n$ times}}.
$$
}
\condinc{
  A \emph{judgement} is an expression 
  $\conone\vdash\termone:\typeone$, where $\conone$ is a linear  environment, 
  $\termone$ is a term, and $\typeone$ is a type in $\QL$. 
}{
Judgements are defined from a linear notion of \emph{environment}.
\begin{varitemize}
\item  
  A \textit{linear environment} $\conone$ is a (possibly
  empty) finite set of assignments in the form $\varone:\typeone$. We
  impose that in a linear environment, each variable $\varone$ 
  occurs \emph{at most} once. 
\item
  If $\conone$ and $\contwo$ are two linear environments assigning types
  to distinct sets of variables, $\conone,\contwo$ is their union.
\item 
  A \emph{judgement} is an expression 
  $\conone\vdash\termone:\typeone$, where $\conone$ is a linear  environment, 
  $\termone$ is a term, and $\typeone$ is a type in $\QL$. 
\end{varitemize}}
\emph{Typing rules} are in Figure~\ref{fig:TypR}.
\begin{figure}
{\scriptsize
\fbox{
\begin{minipage}{1.05\textwidth}
  $$
    \urule{}{\jd{\typet{x}{A}}{\typet{x}{A}}}{\mathsf{(a_v)}}\quad
    \urule{}{\jd{\emcon}{\typet{\ket{0}}{\BB}}}{\mathsf{(a_{q0})}}\quad
    \urule{}{\jd{\emcon}{\typet{\ket{1}}{\BB}}}{\mathsf{(a_{q1})}}\quad
    \urule{}{\jd{\emcon}{U:\BB^{n}\linmap\BB^{n}}}{(\mathsf{a_U})}\quad
    \urule
    {\jd{\conone,\varone:\typeone}{\typet{\termone}{\typetwo}}}
    {\jd{\conone}{\typet{\abstr{\varone}{\termone}}{\typeone\linmap\typetwo}}}
    {(\mathsf{I}_\linmap^1)}
   $$
   $$
    \urule
    {\jd{\conone,\varone:\typeone,\vartwo:\typetwo}{\typet{M}{\typethree}}}
    {\jd{\conone}{\typet{\abstr{\langle\varone,\vartwo\rangle}{\termone}}
        {(\typeone\otimes\typetwo)\linmap\typethree}}}
    {(\mathsf{I}_\linmap^2)}
    \quad
    \brule
    {\jd{\conone}{\typet{M}{A\linmap B}}}
    {\jd{\contwo}{N: A}}
    {\jd{\conone,\contwo}{MN:B}}
    {(\mathsf{E}_\linmap)}
    \quad
    \brule
    {\jd{\conone}{\typet{\termone}{\typeone}}}
    {\jd{\contwo}{\typet{\termtwo}{\typetwo}}}
    {\jd{\conone,\contwo}{\typet{\termone\otimes\termtwo}{\typeone\otimes\typetwo}}}
    {(\mathsf{I}_\otimes)}
    $$
  \vspace{4pt}
\end{minipage}}
\vspace{-2ex}
}
\caption{Typing Rules}\label{fig:TypR}
\end{figure}
Observe that contexts are treated multiplicatively and, as a consequence, variables always
appear exactly once in terms. In other words, a \emph{strictly linear type discipline} is enforced.
\condinc{
\begin{example}[EPR States]\label{ex:epr}
Consider the term 
$\termone_{\mathit{EPR}}=\abstr{\pair{\varone}{\vartwo}}{\mathit{CNOT}(\pairtens{\mathit{H}\varone}{y})}$.
$\termone_{\mathit{EPR}}$ encodes the quantum circuit which takes two input qubits and returns an 
\emph{entangled} state (a quantum state that cannot be expressed as the tensor product of single qubits). 
It can be given the type $\BB\otimes\BB\linmap\BB\otimes\BB$ in the empty context. The following is 
a type derivation $\tdone_{\mathit{EPR}}$ for it, where we mark different occurrences of $\BB$ with distinct indexes. 
(Indexed occurrences will be exploited in Section~\ref{sec:QLTM}, Example~\ref{ex:runIAMQ}.)

{\scriptsize
$$
\urule
{
\brule
{\jd{\emcon}{\mathit{CNOT}}:\BB_{9}\otimes\BB_{10}\linmap\BB_{11}\otimes\BB_{12}}
{
\brule
{
\brule
{\jd{\emcon}{\mathit{H}}:\BB_{21}\linmap\BB_{22}}
{\jd{\varone:\BB_{23}}{\varone:\BB_{24}}}
{\jd{\varone:\BB_{17}}{\mathit{H}\varone:\BB_{18}}}
{(\mathsf{E}_\linmap)}
}
{\jd{\vartwo:\BB_{19}}{\vartwo:\BB_{20}}}
{\jd{\varone:\BB_{13},\vartwo:\BB_{14}}{\pairtens{\mathit{H}\varone}{\vartwo}}:\BB_{15}\otimes\BB_{16}}
{(\mathsf{I}_\otimes)}
}
{\jd{\varone:\BB_{5},\vartwo:\BB_{6}}{\mathit{CNOT}(\pairtens{\mathit{H}\varone}{\vartwo}):\BB_{7}\otimes\BB_{8}}}
{(\mathsf{E}_\linmap)}
}
{
\jd{\emcon}{\termone_{\mathit{EPR}}}:\BB_{1}\otimes\BB_{2}\linmap\BB_{3}\otimes\BB_{4}
}
{(\mathsf{I}_\linmap^2)}
$$
}
\comment{
{\scriptsize
$$
\urule
{
\brule
{\jd{\emcon}{\mathit{CNOT}}:\BB\otimes\BB\linmap\BB\otimes\BB}
{
\brule
{
\brule
{\jd{\emcon}{\mathit{H}}:\BB\linmap\BB}
{\jd{\varone:\BB}{\varone:\BB}}
{\jd{\varone:\BB}{\mathit{H}\varone:\BB}}
{(\mathsf{E}_\linmap)}
}
{\jd{\vartwo:\BB}{\vartwo:\BB}}
{\jd{\varone:\BB,\vartwo:\BB}{\pairtens{\mathit{H}\varone}{\vartwo}}:\BB\otimes\BB}
{(\mathsf{I}_\otimes)}
}
{\jd{\varone:\BB,\vartwo:\BB}{\mathit{CNOT}(\pairtens{\mathit{H}\varone}{\vartwo}):\BB\otimes\BB}}
{(\mathsf{E}_\linmap)}
}
{
\jd{\emcon}{\termone_{\mathit{EPR}}}:\BB\otimes\BB\linmap\BB\otimes\BB
}
{(\mathsf{I}_\linmap^2)}
$$
}
}
$\termone_{\mathit{EPR}}$ and $\tdone_{\mathit{EPR}}$ will be used as running examples in the rest of this
paper, together with the type derivation $\tdtwo_{\mathit{EPR}}$ for 
$\termone_{\mathit{EPR}}(\ket{0}_1\otimes\ket{1}_2)$, which can be easily built from $\tdone_{\mathit{EPR}}$ (see~\cite{EV} for
more details).
\end{example}
}
{

\begin{example}\label{ex:epr}
Consider the following term:
$$
\termone_{\mathit{EPR}}=\abstr{\pair{\varone}{\vartwo}}{\mathit{CNOT}(\pairtens{\mathit{H}\varone}{y})}.
$$
$\termone_{\mathit{EPR}}$ encodes the quantum circuit which takes two input qubits and returns an 
\emph{entangled} state (a quantum state that cannot in general be expressed as the tensor product of single qubits).
It can be given the type $\BB\otimes\BB\linmap\BB\otimes\BB$ in the empty context. Indeed, here is 
a type derivation $\tdone_{\mathit{EPR}}$ for it:
$$
\urule
{
\brule
{\jd{\emcon}{\mathit{CNOT}}:\BB\otimes\BB\linmap\BB\otimes\BB}
{
\brule
{
\brule
{\jd{\emcon}{\mathit{H}}:\BB\linmap\BB}
{\jd{\varone:\BB}{\varone:\BB}}
{\jd{\varone:\BB}{\mathit{H}\varone:\BB}}
{(\mathsf{E}_\linmap)}
}
{\jd{\vartwo:\BB}{\vartwo:\BB}}
{\jd{\varone:\BB,\vartwo:\BB}{\pairtens{\mathit{H}\varone}{\vartwo}}:\BB\otimes\BB}
{(\mathsf{I}_\otimes)}
}
{\jd{\varone:\BB,\vartwo:\BB}{\mathit{CNOT}(\pairtens{\mathit{H}\varone}{\vartwo}):\BB\otimes\BB}}
{(\mathsf{E}_\linmap)}
}
{
\jd{\emcon}{\termone_{\mathit{EPR}}}:\BB\otimes\BB\linmap\BB\otimes\BB
}
{(\mathsf{I}_\linmap^2)}
$$
$\termone_{\mathit{EPR}}$ and $\tdone_{\mathit{EPR}}$ will be used as running examples in the rest of this
paper, together with the following type derivation $\tdtwo_{\mathit{EPR}}$:
$$
\brule
  {\tdone_{\mathit{EPR}}\pof\jd{\emcon}{\termone_{\mathit{EPR}}}:\BB\otimes\BB\linmap\BB\otimes\BB}
  {
    \brule
    {\jd{\emcon}{\ket{0}_1:\BB}}
    {\jd{\emcon}{\ket{1}_2:\BB}}
    {\jd{\emcon}{\ket{0}_1\otimes\ket{1}_2}:\BB\otimes\BB}
    {(\mathsf{I}_\otimes)}
  }
  {\jd{\emcon}{\termone_{\mathit{EPR}}(\ket{0}_1\otimes\ket{1}_2)}:\BB\otimes\BB}
  {(\mathsf{E}_\linmap)}
$$
\end{example}
\condinc{}
{
\comment{
\begin{example}[Deutsch's Circuit]
$$
\mathit{M}_{deu}=\abstr{\pair{\varone}{\vartwo}}{(\abstr{\pair{\varthree}{\varfour}}{\pair{\mathit{H}\varthree}{\varfour}}(\mathit{W}_{f}\pair{\mathit{H}\varone}{\mathit{H}\vartwo}))}
$$

{\scriptsize
$$
\urule
{\brule
{\urule
{
\brule
{
\brule
{\jd{\emcon}{\mathit{H}:\BB\linmap\BB}}
{\jd{\varthree:\BB}{\varthree:\BB}}
{
\jd{\varthree:\BB}{\mathit{H}\varthree:\BB\otimes\BB}
}
{(\mathsf{E}_\linmap)}
}
{
\jd{\varfour:\BB}{\varfour:\BB}
}
{\jd{\varthree:\BB,\varfour:\BB}{\pair{\mathit{H}\varthree}{\varfour}:\BB\otimes\BB}}
{\mathsf{I}_{\otimes}}
}
{
\jd{\emcon}{\abstr{\pair{\varthree}{\varfour}}{\pair{\mathit{H}\varthree}{\varfour}:\BB\otimes\BB\linmap\BB\otimes\BB}}
}
{(\mathsf{I}_\linmap^2)}
}
{\brule
{\jd{\emcon}{\mathit{W}_f:\BB\otimes\BB\linmap\BB\otimes\BB}}
{
\brule
{
\brule
{\jd{\emcon}{\mathit{H}:\BB\linmap\BB}}
{\jd{\varone:\BB}{\varone:\BB}}
{\jd{\varone:\BB}{\mathit{H}\varone:\BB}}
{(\mathsf{E}_\linmap)}
}
{
\brule{}{}{}
{(\mathsf{E}_\linmap)}
}
{\jd{\varone:\BB,\vartwo:\BB}{\pair{\mathit{H}\varone}{\mathit{H}\vartwo}:\BB\otimes\BB}}
{\mathsf{I}_{\otimes}}
}
{\jd{\varone:\BB,\vartwo:\BB}{\mathit{W}_{f}\pair{\mathit{H}\varone}{\mathit{H}\vartwo}:\BB\otimes\BB}}
{(\mathsf{E}_\linmap)}
}
{\jd{\varone:\BB,\vartwo:\BB}{\abstr{\pair{\varthree}{\varfour}}{\pair{\mathit{H}\varthree}{\varfour}}(\mathit{W}_{f}\pair{\mathit{H}\varone}{\mathit{H}\vartwo}):\BB\otimes\BB}}
{(\mathsf{E}_\linmap)}
}
{\jd{\emcon}{\termone_{\mathit{Deu}}}:\BB\otimes\BB\linmap\BB\otimes\BB}
{(\mathsf{I}_\linmap^2)}
$$
}
\end{example}
}
}
If $\tdone\pof\jd{\conone}{(\abstr{\varone}{\termone})\termtwo:\typeone}$, one can build a type derivation $\reduct{\tdone}$ with conclusion 
$\jd{\conone}{\subst{\termone}{\varone}{\termtwo}:\typeone}$ in a canonical way, by going through a constructive subsitution lemma.
Similarly when $\tdone\pof\jd{\conone}{(\abstr{\pair{\varone}{\vartwo}}{\termone})(\pairtens{\termtwo}{\termthree}):\typeone}$.
\begin{lemma}\label{lemma:substlemma}
If $\tdone\pof\jd{\conone,\varone_1:\typeone_1,\ldots,\varone_n:\typeone_n}{\termone:\typetwo}$ and 
for every $1\leq i\leq n$ there is $\tdtwo_i\pof\jd{\contwo_i}{\termtwo_i:\typeone_i}$, then there is a 
canonically defined derivation $\subst{\tdone}{\varone_1,\ldots,\varone_n}{\tdtwo_1,\ldots,\tdtwo_n}$
of $\jd{\conone,\contwo_1,\ldots,\contwo_n}{\subst{\termone}{\varone_1,\ldots,\varone_n}{\termtwo_1,
\ldots,\termtwo_n}:\typetwo}$.
\end{lemma}
\begin{proof}
Just proceed by the usual, simple induction on $\tdone$.
\end{proof}
}

\condinc{}
{The notion of type derivation $\tdone$ of a term $\termone$  and the related definition of  $\reduct{\tdone}$, the type derivation of the reduct of $\termone$,  will be generalized in the following section taking into account quantum superposition.} 
\subsection{An Equational Theory}\label{sect:equtheo}
The $\lambda$-calculus is usually endowed with notions of reduction or equality, both centered
around the $\beta$-rule, according to which a function $\abstr{\varone}{\termone}$ applied
to an argument $\termtwo$ \emph{reduces to} (or \emph{can be considered equal to}) the term
$\termone\{\termtwo/\varone\}$ obtained by replacing all free occurrences of $\varone$ with
$\termtwo$. A reduction relation implicitly provides the underlying calculus with a notion of computation,
while an equational theory is more akin to a reasoning technique. Giving a reduction relation on $\QL$ terms
directly, however, is problematic. What happens when a $n$-ary unitary operator $U$ is faced with an $n$-tuple
of qubits $\ket{b_1\ldots b_n}$? Superposition should somehow arise, but how can we capture it?

In this section, an equational theory for $\QL$ will be introduced. In the next sections, we will prove
that the semantics induced by token machines is \emph{sound} with respect to it. The equational theory we are going
to introduce will be a binary relation on formal, weighted sums of type derivations for $\QL$ terms.
\begin{definition}[Superposed Type Derivation] 
  A \emph{superposed type derivation} of type $(\conone,\typeone)$ is a formal sum 
  \condinc{$\sutone=\sum_{i=1}^n\kappa_i\tdone_i$
  }{
  $$
  \sutone=\sum_{i=1}^n\kappa_i\tdone_i
  $$}
  where for every $1\leq i\leq n$, $\kappa_i\in\CC$ and it holds that $\tdone_i\pof\jd{\conone}{\typet{\termone_i}{\typeone}}$.
  In this case, we write $\jd{\conone}{\sutone:\typeone}$.
  Superposed type derivations will be denoted by metavariables like $\sutone$ or $\suttwo$.
\end{definition}
\condinc{
If $\tdone\pof\jd{\conone}{\termone:\typeone}$ and $\termone$ is a redex (i.e., either $U\ket{\bitone_1\ldots\bitone_k}$
or $(\abstr{\varone}{\termone})\termtwo$ or $(\abstr{\pair{\varone}{\vartwo}}{\termone})$ $(\pairtens{\termtwo}{\termthree})$),
one can easily define a (superposed) type derivation $\reduct{\tdone}$ for the term(s) obtained by firing the redex. 
The binary relation $\eqterm$ on superposed type derivations is defined
as the smallest equivalence relation including all pairs in the form $(\tdone,\reduct{\tdone})$ and closed under all
contexts.
}{		
Please, notice that:		
\begin{varitemize}
\item 
  If $\tdone\pof\jd{\emcon}{U\ket{\bitone_1\ldots\bitone_k}}$, then 
  $\reduct{\tdone}$ is 
  a superposed type derivation in the form $\sum_{x\in B_k}\kappa_x\tdone_x$, where
  $B_k$ is the set of all binary strings of length $k$, $\tdone_x$ is the trivial type derivation
  for $\ket{x}$, and $\kappa_x$ is the complex number corresponding to
  $\ket{x}$ in the vector $\unopone\ket{\bitone_1\ldots\bitone_k}$.
  { \item If $\tdone\pof\jd{\conone}{(\abstr{\varone}{\termone})\termtwo:\typeone}$,  $\reduct{\tdone}$ is the  type derivation with conclusion 
  $\jd{\conone}{\subst{\termone}{\varone}{\termtwo}:\typeone}$ built in a canonical way, by going through a constructive subsitution lemma.
Similarly when $\tdone\pof\jd{\conone}{(\abstr{\pair{\varone}{\vartwo}}{\termone})(\pairtens{\termtwo}{\termthree}):\typeone}$.} 
\item
  All the term constructs can be generalized to operators on superposed type derivations, with the
  proviso that the types match. As an example if $\sutone=\sum_{i}\alpha_{i}\tdone_{i}$ where 
  $\tdone_{i}\pof\jd{\conone}{\typet{\termone_i}{\typeone\linmap\typetwo}}$ and $\tdtwo\pof\jd{\contwo}{\termtwo:\typeone}$,
  $\sutone\tdtwo$ denotes the 
  superposed type derivation $\suttwo=\sum_{i}\alpha_{i}\tdthree_{i}$ where  
  $\tdthree_i\pof\jd{\conone,\contwo}{\termone_i\termtwo:\typetwo}$ and each $\tdthree_i$ is obtained applying the rule 
  $(\mathsf{E}_\linmap)$ to $\tdone_i$ and $\tdtwo$.
\end{varitemize}
A binary relation $\eqterm$ on superposed type derivations having the same type can be given by way of the rules in 
Figure~\ref{fig:eqth}, where we tacitly assume that the involved superposed type derivations have the appropriate 
type whenever needed. Notice that $\eqterm$ is by construction an equivalence relation.
\begin{figure}
  \fbox{
  \begin{minipage}{.97\textwidth}
    \vspace{5pt}
    \begin{center}
      \textbf{Axioms}
      \vspace{-2ex}
    \end{center}
    $$
      \urule{\tdone\pof\jd{\conone}{(\abstr{\pair{\varone}{\vartwo}}{\termone})(\termtwo\otimes\termthree):\typeone}}
         {\tdone\eqterm\reduct{\tdone}}{\mathsf{beta.pair}}
    $$
    \begin{align*}
      \urule{\tdone\pof\jd{\conone}{(\abstr{\varone}{\termone})\termtwo:\typeone}}
         {\tdone\eqterm\reduct{\tdone}}{\mathsf{beta}} &\hspace{25pt}
      \urule{\tdone\pof\jd{\emcon}{U\ket{\bitone_1\ldots\bitone_k}:\BB^k}}{\tdone\eqterm\reduct{\tdone} 
      }{\mathsf{quant}}
    \end{align*}
    \vspace{5pt}
    \begin{center}
      \textbf{Context Closure}
      \vspace{-4ex}
    \end{center}
    \begin{align*}
      \urule
      {\sutone\eqterm\suttwo}
      {\sutone\tdone\eqterm\suttwo\tdone}
      {\mathsf{l.a}}
      &\hspace{1.5pt}&
      \urule
      {\sutone\eqterm\suttwo}
      {\tdone\sutone\eqterm\tdone\suttwo}
      {\mathsf{r.a}}
      &\hspace{1.5pt}&
      \urule
      {\sutone\eqterm\suttwo}
      {\abstr{\varone}{\sutone}\eqterm\abstr{\varone}{\suttwo}}
      {\mathsf{in.}\lambda}
      &\hspace{1.5pt}&
      \urule
      {\sutone\eqterm\suttwo}
      {\abstr{\pair{\varone}{\vartwo}}{\sutone}\eqterm\abstr{\pair{\varone}{\vartwo}}{\suttwo}}
      {\mathsf{in.}\lambda.\mathsf{pair}}
    \end{align*}
    \begin{align*}
      \urule
      {\sutone\eqterm\suttwo}
      {\sutone\otimes\tdone\eqterm\suttwo\otimes\tdone}
      {\mathsf{l.in.tens}}
      & &&
      \urule
      {\sutone\eqterm\suttwo}
      {\tdone\otimes\sutone\eqterm\tdone\otimes\suttwo}
      {\mathsf{r.in.tens}}
      &
      \urule
      {\sutone\eqterm\suttwo}
      {\alpha\sutone+\sutthree\eqterm\alpha\suttwo+\sutthree}
      {\mathsf{sum}}
    \end{align*}
    \vspace{5pt}
    \begin{center}
      \textbf{Reflexive, Symmetric and Transitive Closure}
      \vspace{-4ex}
    \end{center}
    \begin{align*}
      \urule
      {}
      {\sutone\eqterm\sutone}
      {\mathsf{refl}}
      & &&
      \urule
      {\sutone\eqterm\suttwo}
      {\suttwo\eqterm\sutone}
      {\mathsf{sym}}
      &
      \brule
      {\sutone\eqterm\suttwo}
      {\suttwo\eqterm\sutthree}
      {\sutone\eqterm\sutthree}
      {\mathsf{trans}}
    \end{align*}
    \vspace{3pt}
  \end{minipage}}
  \caption{Equational Theory}\label{fig:eqth}
\end{figure}}
When the underlying type derivation is clear from the context, we denote superposed derivations simply by
superposed \emph{terms}. As an example, consider the term $\termone_{\mathit{EPR}}(\ket{0}_1\otimes\ket{1}_2)$
from Example~\ref{ex:epr} and the corresponding type derivation $\tdtwo_{\mathit{EPR}}$ for it \condinc{(see also~\cite{EV})}{}. It is 
convenient to be able to reason as follows, directly on the former:

\condinc{
{\scriptsize
$$
\termone_{\mathit{EPR}}(\ket{0}\otimes\ket{1})\eqterm\mathit{CNOT}(\pairtens{\mathit{H}\ket{0}}{\ket{1}})
     \eqterm\frac{1}{\sqrt{2}}\mathit{CNOT}(\pairtens{\ket{0}}{\ket{1}})+\frac{1}{\sqrt{2}}\mathit{CNOT}(\pairtens{\ket{1}}{\ket{1}})
     $$
     $$
     \eqterm\frac{1}{\sqrt{2}}\pairtens{\ket{0}}{\ket{1}}+\frac{1}{\sqrt{2}}\mathit{CNOT}(\pairtens{\ket{1}}{\ket{1}})
     \eqterm\frac{1}{\sqrt{2}}\pairtens{\ket{0}}{\ket{1}}+\frac{1}{\sqrt{2}}\pairtens{\ket{1}}{\ket{0}}.
$$
}
}
{
\begin{align*}
\termone_{\mathit{EPR}}(\ket{0}\otimes\ket{1})&\eqterm\mathit{CNOT}(\pairtens{\mathit{H}\ket{0}}{\ket{1}})\\
     &\eqterm\frac{1}{\sqrt{2}}\mathit{CNOT}(\pairtens{\ket{0}}{\ket{1}})+\frac{1}{\sqrt{2}}\mathit{CNOT}(\pairtens{\ket{1}}{\ket{1}})\\
     &\eqterm\frac{1}{\sqrt{2}}\pairtens{\ket{0}}{\ket{1}}+\frac{1}{\sqrt{2}}\mathit{CNOT}(\pairtens{\ket{1}}{\ket{1}})\\
     &\eqterm\frac{1}{\sqrt{2}}\pairtens{\ket{0}}{\ket{1}}+\frac{1}{\sqrt{2}}\pairtens{\ket{1}}{\ket{0}}.
\end{align*}
}

Please observe that the equational theory we have just defined can \emph{hardly} be seen as an operational semantics
for $\QL$. Although equations can of course be oriented, it is the very nature of a superposed type derivation which is in principle
problematic from the point of view of quantum computation: what is the mathematical nature of a superposed type derivation? Is it an element
of an Hilbert Space? And if so, of \emph{which one}?
If we consider a simple language such as $\QL$, the questions above may appear overly rhetorical, but we claim they are not.
For example, what would be the quantum meaning of linear beta-reduction? If we want to design beta-reduction according
to the principles of quantum computation, it has to be, at least, easily reversible (unless measurement is implicit in it). 
Moving towards more expressive languages, this non-trivial issue becomes more difficult and a number of 
constraints have to be imposed (for example, superposition of terms can be allowed, but only between ``homogenous'' terms, i.e. 
terms which have an identical skeleton~\cite{vT04}).  This is the reason for which promising calculi~\cite{vT04} fail to be 
canonical models for quantum programming languages. 
This issue has been faced in literature without satisfactory answers, yielding  a number of convincing arguments in favor of the 
(implicit or explicit) classical control of quantum data~\cite{mscs2009,SV06}. 
\newcommand{\axrls}{\mathsf{AX}}
\newcommand{\ccrls}{\mathsf{CC}}
\condinc{}{
\subsubsection{Equational Theory Derivations in Normal Form}
Sometime it is quite useful to assume that a derivation for $\sutone\eqterm\suttwo$ is in a peculiar form, defined by giving an order
on the rules in Figure~\ref{fig:eqth}. More specifically, define the following two sets of rules:
\begin{align*}
\axrls&=\{\mathsf{beta},\mathsf{beta.pair},\mathsf{quant}\};\\
\ccrls&=\{\mathsf{l.a},\mathsf{r.a},\mathsf{in.}\lambda,\mathsf{in.}\lambda\mathsf{.pair},\mathsf{l.in.tens},\mathsf{r.in.tens}\}.
\end{align*}
A derivation of $\sutone\eqterm\suttwo$ is said to be \emph{in normal form} (and we write $\sutone\eqtermnf\suttwo$) iff:
\begin{varitemize}
\item
  either the derivation is obtained by applying rule $\mathsf{refl}$;
\item
  or any branch in the derivation consists in instances of rules from $\axrls$, possibly followed by instances of rules in $\ccrls$, possibly followed
  by instances of $\mathsf{sum}$, possibly followed by instances of $\mathsf{sym}$, possibly followed by instances of $\mathsf{trans}$. 
\end{varitemize}
In other words, a derivation of $\sutone\eqterm\suttwo$ is in normal form iff rules are applied in a certain order. As an example, we
cannot apply transitivity or symmetry closure rules too early, i.e., before context closure rules. One may wonder whether this restricts the
class of provable equivalences. Infact it does not:
\begin{proposition}\label{prop:normalform}
$\sutone\eqterm\suttwo$ iff $\sutone\eqtermnf\suttwo$.
\end{proposition}
\begin{proof}
If $\sutone\eqtermnf\suttwo$, then of course $\sutone\eqterm\suttwo$.
The converce can be showed by induction on the \emph{height} $n$ of a proof of $\sutone\eqterm\suttwo$, enriching the thesis by
prescribing that the \emph{height} of the obtained proof of $\sutone\eqterm\suttwo$ must be at most $n$:
\begin{varitemize}
\item
  If $\sutone\eqterm\suttwo$ is proved by rules in $\axrls$ or by $\mathsf{refl}$, then 
  by definition $\sutone\eqtermnf\suttwo$.
\item
  If $\sutone\eqterm\suttwo$ is derived by rules in $\ccrls$ from a proof $\tdone$, then:
  \begin{varitemize}
  \item
    If the rules in $\tdone$ are all from $\axrls$ and $\ccrls$, then there is nothing to do.
  \item
    If the last rule in $\tdone$ is $\mathsf{sum}$, then we can apply one of the following transformations,
    so as to be able to apply the induction hypothesis:

    {\footnotesize
    $$
      \urule
      {
        \urule
        {\sutthree\eqterm\sutfive}
        {\alpha\sutthree+\sutfour\eqterm\alpha\sutfive+\sutfour}
        {\mathsf{sum}}
      }
      {\alpha\sutthree\tdone+\sutfour\tdone\eqterm\alpha\sutfive\tdone+\sutfour\tdone}
      {\mathsf{l.a}}
      \qquad
      \Longrightarrow
      \qquad
      \urule
      {
        \urule
        {\sutthree\eqterm\sutfive}
        {\sutthree\tdone\eqterm\sutfive\tdone}
        {\mathsf{\mathsf{l.a}}}
      }
      {\alpha\sutthree\tdone+\sutfour\tdone\eqterm\alpha\sutfive\tdone+\sutfour\tdone}
      {\mathsf{sum}}
    $$
    \vspace{4pt}
    $$
      \urule
      {
        \urule
        {\sutthree\eqterm\sutfive}
        {\alpha\sutthree+\sutfour\eqterm\alpha\sutfive+\sutfour}
        {\mathsf{sum}}
      }
      {\alpha\tdone\sutthree+\tdone\sutfour\eqterm\alpha\tdone\sutfive+\tdone\sutfour}
      {\mathsf{r.a}}
      \qquad
      \Longrightarrow
      \qquad
      \urule
      {
        \urule
        {\sutthree\eqterm\sutfive}
        {\tdone\sutthree\eqterm\tdone\sutfive}
        {\mathsf{\mathsf{r.a}}}
      }
      {\alpha\tdone\sutthree+\tdone\sutfour\eqterm\alpha\tdone\sutfive+\tdone\sutfour}
      {\mathsf{sum}}
    $$
    \vspace{4pt}
    $$
      \urule
      {
        \urule
        {\sutthree\eqterm\sutfive}
        {\alpha\sutthree+\sutfour\eqterm\alpha\sutfive+\sutfour}
        {\mathsf{sum}}
      }
      {\alpha\abstr{\varone}{\sutthree}+\abstr{\varone}{\sutfour}\eqterm\alpha\abstr{\varone}{\sutfive}+\abstr{\varone}{\sutfour}}
      {\mathsf{in.}\lambda}
      \qquad
      \Longrightarrow
      \qquad
      \urule
      {
        \urule
        {\sutthree\eqterm\sutfive}
        {\abstr{\varone}{\sutthree}\eqterm\abstr{\varone}{\sutfive}}
        {\mathsf{\mathsf{in.}\lambda}}
      }
      {\alpha\abstr{\varone}{\sutthree}+\abstr{\varone}{\sutfour}\eqterm\alpha\abstr{\varone}{\sutfive}+\abstr{\varone}{\sutfour}}
      {\mathsf{sum}}
    $$
    \vspace{4pt}
    $$
      \urule
      {
        \urule
        {\sutthree\eqterm\sutfive}
        {\alpha\sutthree+\sutfour\eqterm\alpha\sutfive+\sutfour}
        {\mathsf{sum}}
      }
      {\alpha\abstr{\pair{\varone}{\vartwo}}{\sutthree}+\abstr{\pair{\varone}{\vartwo}}{\sutfour}\eqterm\alpha\abstr{\pair{\varone}{\vartwo}}{\sutfive}+\abstr{\pair{\varone}{\vartwo}}{\sutfour}}
      {\mathsf{in.}\lambda\mathsf{.pair}}
      \qquad
      \Longrightarrow
    $$
    $$
      \urule
      {
        \urule
        {\sutthree\eqterm\sutfive}
        {\abstr{\pair{\varone}{\vartwo}}{\sutthree}\eqterm\abstr{\pair{\varone}{\vartwo}}{\sutfive}}
        {\mathsf{\mathsf{in.}\lambda}\mathsf{.pair}}
      }
      {\alpha\abstr{\pair{\varone}{\vartwo}}{\sutthree}+\abstr{\pair{\varone}{\vartwo}}{\sutfour}\eqterm\alpha\abstr{\pair{\varone}{\vartwo}}{\sutfive}+\abstr{\pair{\varone}{\vartwo}}{\sutfour}}
      {\mathsf{sum}}
    $$
    \vspace{4pt}
    $$
      \urule
      {
        \urule
        {\sutthree\eqterm\sutfive}
        {\alpha\sutthree+\sutfour\eqterm\alpha\sutfive+\sutfour}
        {\mathsf{sum}}
      }
      {\alpha\sutthree\otimes\tdone+\sutfour\otimes\tdone\eqterm\alpha\sutfive\otimes\tdone+\sutfour\otimes\tdone}
      {\mathsf{l.in.tens}}
      \qquad
      \Longrightarrow
      \qquad
      \urule
      {
        \urule
        {\sutthree\eqterm\sutfive}
        {\sutthree\otimes\tdone\eqterm\sutfive\otimes\tdone}
        {\mathsf{\mathsf{l.in.tens}}}
      }
      {\alpha\sutthree\otimes\tdone+\sutfour\tdone\otimes\eqterm\alpha\sutfive\otimes\tdone+\sutfour\otimes\tdone}
      {\mathsf{sum}}
    $$
    \vspace{4pt}
    $$
      \urule
      {
        \urule
        {\sutthree\eqterm\sutfive}
        {\alpha\sutthree+\sutfour\eqterm\alpha\sutfive+\sutfour}
        {\mathsf{sum}}
      }
      {\alpha\tdone\otimes\sutthree+\tdone\otimes\sutfour\eqterm\alpha\tdone\otimes\sutfive+\tdone\otimes\sutfour}
      {\mathsf{r.in.tens}}
      \qquad
      \Longrightarrow
      \qquad
      \urule
      {
        \urule
        {\sutthree\eqterm\sutfive}
        {\tdone\otimes\sutthree\eqterm\tdone\otimes\sutfive}
        {\mathsf{\mathsf{r.in.tens}}}
      }
      {\alpha\tdone\otimes\sutthree+\tdone\otimes\sutfour\eqterm\alpha\tdone\otimes\sutfive+\tdone\otimes\sutfour}
      {\mathsf{sum}}
    $$}
  \item
    If the last rule in $\tdone$ is $\mathsf{sym}$ or $\mathsf{trans}$, then we can easily apply similar transformations,
    so as to be able to apply the induction hypothesis.
  \item
    If the last rule in $\tdone$ is $\mathsf{refl}$, then we can derive $\sutone\eqterm\suttwo$ by a single application
    of $\mathsf{refl}$.
  \end{varitemize}
\item
  If $\sutone\eqterm\suttwo$ is derived by $\mathsf{sum}$ from a proof $\tdone$, then:
  \begin{varitemize}
  \item
    If the rules in $\tdone$ are all from $\axrls$ or $\ccrls$, or are $\mathsf{sum}$, then there is nothing to do.
  \item
    If the last rule in $\tdone$ is $\mathsf{sym}$, then we can apply the following transformation,
    so as to be able to apply the induction hypothesis:
    $$
      \urule
      {
        \urule
        {\sutthree\eqterm\sutfive}
        {\sutfive\eqterm\sutthree}
        {\mathsf{sym}}
      }
      {\alpha\sutfive+\sutfour\eqterm\alpha\sutthree+\sutfour}
      {\mathsf{sum}}
      \qquad
      \Longrightarrow
      \qquad
      \urule
      {
        \urule
        {\sutthree\eqterm\sutfive}
        {\alpha\sutthree+\sutfour\eqterm\alpha\sutfive+\sutfour}
        {\mathsf{sum}}
      }
      {\alpha\sutfive+\sutfour\eqterm\alpha\sutthree+\sutfour}
      {\mathsf{sym}}
    $$
  \item
    If the last rule in $\tdone$ is $\mathsf{trans}$, then we can apply the following transformation,
    so as to be able to apply the induction hypothesis
    
    {\footnotesize
    $$
    \urule
    {
      \brule
      {\sutthree\eqterm\sutfive}
      {\sutfive\eqterm\sutsix}
      {\sutthree\eqterm\sutsix}
      {\mathsf{trans}}
    }
    {\alpha\sutthree+\sutfour\eqterm\alpha\sutsix+\sutfour}
    {\mathsf{sum}}
    \qquad
    \Longrightarrow
    \qquad
    \brule
    {
      \urule
      {\sutthree\eqterm\sutfive}
      {\alpha\sutthree+\sutfour\eqterm\alpha\sutfive+\sutfour}
      {\mathsf{sum}}
    }
    {
      \urule
      {\sutfive\eqterm\sutsix}
      {\alpha\sutfive+\sutfour\eqterm\alpha\sutsix+\sutfour}
      {\mathsf{sum}}
    }
    {\alpha\sutfive+\sutfour\eqterm\alpha\sutsix+\sutfour}
    {\mathsf{trans}}
    $$}
  \item
    If the last rule in $\tdone$ is $\mathsf{refl}$, then we can derive $\sutone\eqterm\suttwo$ by a single application
    of $\mathsf{refl}$.
  \end{varitemize}
\item
  If $\sutone\eqterm\suttwo$ is derived by $\mathsf{sym}$ from a proof $\tdone$, then:
  \begin{varitemize}
  \item
    If the rules in $\tdone$ are all from $\axrls$ or $\ccrls$, or are $\mathsf{sum}$ or $\mathsf{sym}$, then there is nothing to do.
  \item
    If the last rule in $\tdone$ is $\mathsf{trans}$, then we can apply the following transformation,
    so as to be able to apply the induction hypothesis:
    $$
    \urule
    {
      \brule
      {\sutthree\eqterm\sutfive}
      {\sutfive\eqterm\sutsix}
      {\sutthree\eqterm\sutsix}
      {\mathsf{trans}}
    }
    {\sutsix\eqterm\sutthree}
    {\mathsf{sym}}
    \qquad
    \Longrightarrow
    \qquad
    \brule
    {
      \urule
      {\sutfive\eqterm\sutsix}
      {\sutsix\eqterm\sutfive}
      {\mathsf{sym}}
    }
    {
      \urule
      {\sutthree\eqterm\sutfive}
      {\sutfive\eqterm\sutthree}
      {\mathsf{sym}}
    }
    {\sutsix\eqterm\sutthree}
    {\mathsf{trans}}
    $$
  \item
    If the last rule in $\tdone$ is $\mathsf{refl}$, then we can derive $\sutone\eqterm\suttwo$ by a single application
    of $\mathsf{refl}$.
  \end{varitemize}
\item
  If $\sutone\eqterm\suttwo$ is derived by $\mathsf{trans}$ from two proofs of $\tdone$ and $\tdtwo$, then
  if either $\tdone$ or $\tdtwo$ is derived by $\mathsf{refl}$, then the required proof is already in our
  hand. Otherwise, there is nothing to do.
\end{varitemize}
This concludes the proof.
\end{proof}
}
\section{A Token Machine for $\QL$}\label{sec:QLTM}
In this section we describe an interpretation of $\QL$ type derivations 
in terms of a specific token machine called $\IAM{\QL}$. 

With a slight abuse of notation, a permutation $\permone:\{1,\ldots,n\}\rightarrow\{1,\ldots,n\}$ will
be often applied to sequences of length $n$ with the obvious meaning:
$\permone(a_1,\ldots,a_n)=a_{\permone(1)},\ldots,a_{\permone(n)}$.
Similarly, such a permutation can be seen as the \emph{unique} unitary operator on $\CC^{2^n}$ which
sends $\ket{b_1\ldots b_n}$ to $\ket{b_{\permone(1)}\ldots b_{\permone(n)}}$.
Suppose given an operator $\unopone\in\uset$ of arity $n\in\NN$. Now, take
a natural number $m\geq n$ and $n$ distinct natural numbers $j_1,\ldots,j_n$,
all of them smaller or equal to $m$. With $\unopone_m^{j_1,\ldots,j_n}$ (or simply
with $\unopone^{j_1,\ldots,j_n}$) we indicate the operator of arity $m$ which acts
like $\unopone$ on the qubits indexed with $j_1,\ldots,j_n$ and leave all the other
qubits unchanged.

In the following, with a slight abuse of notation,
occurrences of types in type derivations are confused with
types themselves. On the other hand, occurrences of types
\emph{inside other types} will be defined quite precisely, as follows.
\emph{Contexts} (types with an hole) are denoted by metavariables
like $\ctone,\cttwo$. A context $\ctone$ is said to be
\emph{a context for a type $\typeone$} if $\ctone[\BB]=\typeone$.
\emph{Negative contexts} (i.e., contexts where the hole is in negative
position) are denoted by metavariables like $\nconone,\ncontwo$. 
\emph{Positive} ones are denoted by metavariables like $\pconone,\pcontwo$.
An \emph{occurrence} of $\BB$ in the type derivation $\tdone$
is a pair $(\typeone,\ctone)$, where $\typeone$ is an occurrence
of a type in $\tdone$ and $\ctone$ is a context for $\typeone$.
Sequences of occurrences are indicated with metavariables like
$\soccone,\socctwo$ (possibly indexed). All sequences of occurrences we will deal with
do not contain duplicates. Type constructors $\linmap$ and $\tens$ can
be generalized to operators on occurrences and sequences of occurrences, e.g.
$(\typeone,\ctone)\linmap\typetwo$ is just $(\typeone\linmap\typetwo,\ctone\linmap\typetwo)$.

Given (an occurrence of) a type $\typeone$, all positive
and negative occurrences of $\BB$ inside $\typeone$ can be
put in sequences called $\poccs{\typeone}$ and $\noccs{\typeone}$
\condinc{\hspace{-0.5ex}.
}
{
as follows (where $\cdot$ is sequence concatenation):
\begin{align*}
  \poccs{\BB}&=(\BB,\emct);\\
  \noccs{\BB}&=\emseq;\\
  \poccs{\typeone\tens\typetwo}&=(\poccs{\typeone}\tens\typetwo)\cdot(\typeone\tens\poccs{\typetwo});\\
  \noccs{\typeone\tens\typetwo}&=(\noccs{\typeone}\tens\typetwo)\cdot(\typeone\tens\noccs{\typetwo});\\
  \poccs{\typeone\linmap\typetwo}&=(\noccs{\typeone}\linmap\typetwo)\cdot(\typeone\linmap\poccs{\typetwo});\\
  \noccs{\typeone\linmap\typetwo}&=(\poccs{\typeone}\linmap\typetwo)\cdot(\typeone\linmap\noccs{\typetwo}).
\end{align*}
}
As an example, the positive occurrences in the type $\BB\linmap\BB\tens\BB$ should be the two
rightmost ones. And indeed:
\condinc
{
$ \poccs{\BB\linmap\BB\tens\BB}=(\BB,\BB\linmap(\emct\tens\BB)),(\BB,\BB\linmap(\BB\tens\emct))$.
}
{
\begin{align*}
  \poccs{\BB\linmap\BB\tens\BB}&=(\noccs{\BB}\linmap\BB\tens\BB)\cdot(\BB\linmap\poccs{\BB\tens\BB})\\
     &=\emseq\cdot(\BB\linmap\poccs{\BB\tens\BB})=\BB\linmap\poccs{\BB\tens\BB}\\
     &=(\BB\linmap(\poccs{\BB}\tens\BB))\cdot(\BB\linmap(\BB\tens\poccs{\BB}))\\
     &=(\BB,\BB\linmap(\emct\tens\BB)),(\BB,\BB\linmap(\BB\tens\emct)).
\end{align*}

}
For every type derivation $\tdone$, $\bitocc{\tdone}$ is the sequence of all occurrences
of $\BB$ in $\tdone$ which are introduced by the rules $(\mathsf{a_{q0}})$ and $(\mathsf{a_{q1}})$ (from Figure~\ref{fig:TypR}). Similarly,
$\bitval{\tdone}$ is the corresponding sequence of binary digits, seen as a vector in $\CC^{2^{|\bitocc{\tdone}|}}$. 
Both in $\bitocc{\tdone}$ and in $\bitval{\tdone}$, the order is the one induced by the natural number labeling 
the underlying bit in $\tdone$. 
\condinc{}{
As an example, consider the following type derivation, and call it $\tdone$:
$$
\brule
    {\jd{\emcon}{\ket{0}_2:\BB_1}}
    {\jd{\emcon}{\ket{1}_1:\BB_2}}
    {\jd{\emcon}{\ket{0}_2\otimes\ket{1}_1}:\BB_3\otimes\BB_4}
    {(\mathsf{I}_\otimes)}
$$
There are four occurrences of $\BB$ in it, and we have indexed it with the first four positive natural numbers, just
to be able to point at them without being forced to use the formal, context machinery. Only two of them, namely the
upper ones, are introduced by instances of the rules $(\mathsf{a_{q0}})$ and $(\mathsf{a_{q1}})$. Moreover, the rightmost
one serves to type a bit having an index (namely $1$) greater than the one in the other instance (namely $2$). As a consequence,
$\bitocc{\tdone}$ is the sequence $\BB_2,\BB_1$. The two instances introduces bits $0$ and $1$; then $\bitval{\tdone}=\ket{1}\otimes\ket{0}$.
As another example, one can easily compute $\bitocc{\tdone_{\mathit{EPR}}}$ and $\bitval{\tdone_{\mathit{EPR}}}$ (where
$\tdone_{\mathit{EPR}}$ is from Example~\ref{ex:epr}), finding out that both are the empty sequence.}

Finally, we are able to define, for every $\tdone$, the abstract machine $\autom{\tdone}$ interpreting it:
\begin{varitemize}
\item
  The \emph{states} of $\autom{\tdone}$ form a set $\states{\tdone}$ and are in the form
  $(\occone_1,\ldots,\occone_n,\qrone)$
  where:
  \begin{varitemize}
  \item
    $\occone_1,\ldots,\occone_n$ are occurrences of the 
    type $\BB$ in $\tdone$; 
  \item
    $\qrone$ is a \emph{quantum register} on $n$ qubits, i.e. a normalized vector in $\CC^{2^n}$\condinc{}{(see Section~\ref{sec:qcomp})}.
  \end{varitemize}
\item
  The \emph{transition relation} $\trans{\tdone}\subseteq\states{\tdone}\times\states{\tdone}$
  is defined based on $\tdone$, following Figure~\ref{fig:transone} and Figure~\ref{fig:transtwo}.
  In the latter, 
  each of the $2n$ occurrences of $\BB$ in the type of $U$ is simply denoted through
  its index, and for every $1\leq k\leq m$, $i_k$ is the position of $\BB_k$ in the sequence
  $(\soccone_1,\BB_{j_1},\soccone_2,\ldots,\soccone_{m},\BB_{j_m},\soccone_{m+1})$. \condinc{In the former 
  the transition rules induced by $(\mathsf{I}_\linmap^2)$ have been
  elided for the sake of simplicity (see~\cite{EV}).}{}
\end{varitemize}
The number of positive (negative, respectively) occurrences
of $\BB$ in the conclusion of $\tdone$ is said to be the \emph{output arity}
(the \emph{input arity}, respectively) of $\tdone$.
\begin{figure}
\fbox{
\begin{minipage}{.97\textwidth}
{\scriptsize
\begin{center}
\begin{tabular}{cc}
\begin{minipage}{2cm}
$$
\urule{}{\jd{\typet{x}{\typeone_1}}{\typet{x}{\typeone_2}}}{}
$$
\end{minipage}
&
\begin{minipage}{6.5cm}
\begin{eqnarray*}
((\soccone,(\typeone_1,\pconone),\socctwo),\qrone) &\trans{\tdone}& 
  ((\soccone,(\typeone_2,\pconone),\socctwo),\qrone)\\
((\soccone,(\typeone_2,\nconone),\socctwo),\qrone) &\trans{\tdone}& 
  ((\soccone,(\typeone_1,\nconone),\socctwo),\qrone)
\end{eqnarray*}
\end{minipage}
\end{tabular}

\vspace{2pt}

\begin{tabular}{cc}
\begin{minipage}{3.5cm}
$$
\urule
  {\jd{\conone_1,\varone:\typeone_1}{\typet{\termone}{\typetwo_1}}}
  {\jd{\conone_2}{\typet{\abstr{\varone}{\termone}}{\typeone_2\linmap\typetwo_2}}}
  {}
$$
\end{minipage}
&
\begin{minipage}{7cm}
$$
\begin{array}{c}
((\soccone,(\typeone_1,\nconone),\socctwo),\qrone) \trans{\tdone} 
  ((\soccone,(\typeone_2\linmap\typetwo_2,\nconone\linmap\typetwo_2),\socctwo),\qrone)\\ \vspace{-7pt} \\
((\soccone,(\typeone_2\linmap\typetwo_2,\pconone\linmap\typetwo_2),\socctwo),\qrone) \trans{\tdone}
  ((\soccone,(\typeone_1,\pconone),\socctwo),\qrone)\\ \vspace{-7pt} \\
((\soccone,(\typetwo_1,\pconone),\socctwo),\qrone) \trans{\tdone} 
  ((\soccone,(\typeone_2\linmap\typetwo_2,\typeone_2\linmap\pconone),\socctwo),\qrone)\\ \vspace{-7pt} \\
((\soccone,(\typeone_2\linmap\typetwo_2,\typeone_2\linmap\nconone),\socctwo),\qrone) \trans{\tdone} 
  ((\soccone,(\typetwo_1,\nconone),\socctwo),\qrone)\\ \vspace{-7pt} \\
((\soccone,(\conone_2,\pconone),\socctwo),\qrone) \trans{\tdone}
  ((\soccone,(\conone_1,\pconone),\socctwo),\qrone)\\ \vspace{-7pt} \\
((\soccone,(\conone_1,\nconone),\socctwo),\qrone) \trans{\tdone} 
  ((\soccone,(\conone_2,\nconone),\socctwo),\qrone)\\ \vspace{-7pt} \\
\end{array}
$$
\end{minipage}
\end{tabular}

\vspace{2pt}

\condinc{}{
\begin{tabular}{cc}
\begin{minipage}{4.5cm}
$$
\urule
   {\jd{\conone_1,\varone:\typeone_1,\vartwo:\typetwo_1}{\typet{M}{\typethree_1}}}
   {\jd{\conone_2}{\typet{\abstr{\langle\varone,\vartwo\rangle}{\termone}}
          {(\typeone_2\otimes\typetwo_2)\linmap\typethree_2}}}
   {}
$$
\end{minipage}
&
\begin{minipage}{8cm}
  $$
  \begin{array}{c}
    ((\soccone,(\typeone_1,\nconone),\socctwo),\qrone) \trans{\tdone} 
      ((\soccone,(\typeone_2\otimes\typetwo_2\linmap\typethree_2,\nconone\otimes\typetwo_2\linmap\typethree_2),\socctwo),\qrone)\\ \vspace{-7pt} \\
    ((\soccone,(\typeone_2\otimes\typetwo_2\linmap\typethree_2,\pconone\otimes\typetwo_2\linmap\typethree_2),\socctwo),\qrone) \trans{\tdone}
       ((\soccone,(\typeone_1,\pconone),\socctwo),\qrone)\\ \vspace{-7pt} \\
    ((\soccone,(\typetwo_1,\nconone),\socctwo),\qrone) \trans{\tdone} 
      ((\soccone,(\typeone_2\otimes\typetwo_2\linmap\typethree_2,\typeone_2\otimes\nconone\linmap\typetwo_2),\socctwo),\qrone)\\ \vspace{-7pt} \\
    ((\soccone,(\typeone_2\otimes\typetwo_2\linmap\typethree_2,\typeone_2\otimes\pconone\linmap\typethree_2),\socctwo),\qrone) \trans{\tdone}
       ((\soccone,(\typetwo_1,\pconone),\socctwo),\qrone)\\ \vspace{-7pt} \\
    ((\soccone,(\typethree_1,\pconone),\socctwo),\qrone) \trans{\tdone} 
       ((\soccone,(\typeone_2\otimes\typetwo_2\linmap\typethree_2,\typeone_2\otimes\typetwo_2\linmap\pconone),\socctwo),\qrone)\\ \vspace{-7pt} \\
    ((\soccone,(\typeone_2\otimes\typetwo_2\linmap\typethree_2,\typeone_2\otimes\typetwo_2\linmap\nconone),\socctwo),\qrone) \trans{\tdone} 
       ((\soccone,(\typethree_1,\nconone),\socctwo),\qrone)\\ \vspace{-7pt} \\
    ((\soccone,(\conone_2,\pconone),\socctwo),\qrone) \trans{\tdone}
    ((\soccone,(\conone_1,\pconone),\socctwo),\qrone)\\ \vspace{-7pt} \\
    ((\soccone,(\conone_1,\nconone),\socctwo),\qrone) \trans{\tdone} 
    ((\soccone,(\conone_2,\nconone),\socctwo),\qrone)\\ \vspace{-7pt} \\
  \end{array}
  $$
\end{minipage}
\end{tabular}

\vspace{2pt}}

\begin{tabular}{cc}
\begin{minipage}{4.5cm}
$$
\brule
    {\jd{\conone_1}{\typet{\termone}{\typeone_1\linmap\typetwo_1}}}
    {\jd{\contwo_1}{\termtwo:\typeone}_2}
    {\jd{\conone_2,\contwo_2}{\termone\termtwo:\typetwo_2}}
    {}
$$
\end{minipage}
&
\begin{minipage}{7cm}
$$
\begin{array}{c}
((\soccone,(\typeone_2,\pconone),\socctwo),\qrone) \trans{\tdone} 
  ((\soccone,(\typeone_1\linmap\typetwo_1,\pconone\linmap\typetwo_1),\socctwo),\qrone)\\ \vspace{-7pt} \\
((\soccone,(\typeone_1\linmap\typetwo_1,\nconone\linmap\typetwo_1),\socctwo),\qrone) \trans{\tdone} 
  ((\soccone,(\typeone_2,\nconone),\socctwo),\qrone)\\ \vspace{-7pt} \\
((\soccone,(\typeone_1\linmap\typetwo_1,\typeone_1\linmap\pconone),\socctwo),\qrone) \trans{\tdone} 
  ((\soccone,(\typetwo_2,\pconone),\socctwo),\qrone)\\ \vspace{-7pt} \\
((\soccone,(\typetwo_2,\nconone),\socctwo),\qrone) \trans{\tdone} 
  ((\soccone,(\typeone_1\linmap\typetwo_1,\typeone\linmap\nconone),\socctwo),\qrone)\\ \vspace{-7pt} \\
((\soccone,(\conone_2,\pconone),\socctwo),\qrone) \trans{\tdone} 
  ((\soccone,(\conone_1,\pconone),\socctwo),\qrone)\\ \vspace{-7pt} \\
((\soccone,(\conone_1,\nconone),\socctwo),\qrone) \trans{\tdone} 
  ((\soccone,(\conone_2,\nconone),\socctwo),\qrone)\\ \vspace{-7pt} \\
((\soccone,(\contwo_2,\pconone),\socctwo),\qrone) \trans{\tdone} 
  ((\soccone,(\contwo_1,\pconone),\socctwo),\qrone)\\ \vspace{-7pt} \\
((\soccone,(\contwo_1,\nconone),\socctwo),\qrone) \trans{\tdone} 
  ((\soccone,(\contwo_2,\nconone),\socctwo),\qrone)
\end{array}
$$
\end{minipage}
\end{tabular}

\vspace{2pt}

\begin{tabular}{cc}
\begin{minipage}{4.25cm}
$$
\brule
    {\jd{\conone_1}{\typet{\termone}{\typeone_1}}}
    {\jd{\contwo_1}{\typet{\termtwo}{\typetwo_1}}}
    {\jd{\conone_2,\contwo_2}{\typet{\termone\otimes\termtwo}{\typeone_2\otimes\typetwo_2}}}
    {(\mathsf{I}_\otimes)}
$$
\end{minipage}
&
\begin{minipage}{7.25cm}
$$
\begin{array}{c}
((\soccone,(\typeone_2\otimes\typetwo_2,\nconone\otimes\typetwo_2),\socctwo),\qrone) \trans{\tdone} 
((\soccone,(\typeone_1,\nconone),\socctwo),\qrone)\\ \vspace{-7pt} \\
((\soccone,(\typeone_2\otimes\typetwo_2,\typeone_2\otimes\nconone),\socctwo),\qrone) \trans{\tdone} 
((\soccone,(\typetwo_1,\nconone),\socctwo),\qrone)\\ \vspace{-7pt} \\
((\soccone,(\typeone_1,\pconone),\socctwo),\qrone) \trans{\tdone} 
  ((\soccone,(\typeone_2\otimes\typetwo_2,\pconone\otimes\typetwo_2),\socctwo),\qrone)\\ \vspace{-7pt} \\
((\soccone,(\typetwo_1\pconone),\socctwo),\qrone) \trans{\tdone} 
  ((\soccone,(\typeone_2\otimes\typetwo_2,\typeone_2\otimes\pconone),\socctwo),\qrone)\\ \vspace{-7pt} \\
((\soccone,(\conone_1,\nconone),\socctwo),\qrone) \trans{\tdone} 
  ((\soccone,(\conone_2,\nconone),\socctwo),\qrone)\\ \vspace{-7pt} \\
((\soccone,(\contwo_1,\nconone),\socctwo),\qrone) \trans{\tdone} 
  ((\soccone,(\contwo_2,\nconone),\socctwo),\qrone)\\ \vspace{-7pt} \\
((\soccone,(\conone_2,\pconone),\socctwo),\qrone) \trans{\tdone} 
  ((\soccone,(\conone_1,\pconone),\socctwo),\qrone)\\ \vspace{-7pt} \\
((\soccone,(\contwo_2,\pconone),\socctwo),\qrone) \trans{\tdone} 
  ((\soccone,(\contwo_1,\pconone),\socctwo),\qrone)
\end{array}
$$
\end{minipage}
\end{tabular}
\vspace{10pt}
\end{center}}
\end{minipage}}
\caption{Quantum GoI Machine --- Classical Rules}\label{fig:transone}
\end{figure}
\begin{figure}
\fbox{
{\scriptsize
\begin{minipage}{.97\textwidth}
\begin{center}
$$
\urule{}{\jd{\emcon}{U:\BB_1\otimes\ldots\otimes\BB_m\linmap\BB_{m+1}\otimes\ldots\otimes\BB_{2m}}}{}
$$
$$
\begin{array}{c}
((\soccone_1,\BB_{j_1},\soccone_2,\ldots,\soccone_{m},\BB_{j_m},\soccone_{m+1}),\qrone)\\
\trans{\tdone}\\
((\soccone_1,\BB_{j_1+m},\soccone_2,\ldots,\soccone_{m},\BB_{j_m+m},\soccone_{m+1}),
\unopone^{i_1,\ldots,i_m}(\qrone))
\end{array}
$$
\end{center}
\vspace{7pt}
\end{minipage}
}
}
\caption{Quantum GoI Machine --- Quantum Rules}\label{fig:transtwo}
\end{figure}
Given a type derivation $\tdone$, the relation $\trans{\tdone}$ enjoys a strong form of confluence:
\begin{proposition}[One-step Confluence of $\trans{\tdone}$]
Let $\stone, \sttwo, \stthree\in \states{\tdone}$ be such that $\stone\trans{\tdone} \sttwo$ and $\stone\trans{\tdone}\stthree$. 
Then either $\sttwo=\stthree$ or there exists a state $\stfour$ such that $\sttwo\trans{\tdone} \stfour$ and $\stthree\trans{\tdone} \stfour$.
\end{proposition}
\begin{proof}
By simply inspecting the various rules. Notice that there are no critical pairs in 
$\trans{\tdone}$. \qed
\end{proof}

Suppose, for the sake of simplicity, that $\tdone$ is a type derivation of
$\jd{\emcon}{\termone:\typeone}$.
An \emph{initial state} for $\qrone$ is a state in the
form $(\noccs{\typeone}\cdot\bitocc{\tdone},\qrone\otimes\bitval{\tdone})$.
Given a permutation $\permone$ on $n$ elements, a \emph{final state for} $\qrone$ 
\emph{and} $\permone$ is one in the form $(\soccone,\qrone)$,
where $\soccone=\permone(\poccs{\typeone})$.

\condinc{
Given a type derivation $\tdone$, the partial function \emph{computed by $\tdone$} 
is $\pfun{\tdone}:\CC^{2^n}\rightharpoonup\CC^{2^m}$
(where $n$ and $m$ are the input and output arity of $\tdone$) and is
defined by stipulating that $\pfun{\tdone}(\qrone)=\permone^{-1}(\qrtwo)$ iff
any initial state for $\qrone$ rewrites into a final state for $\qrtwo$ and $\permone$.
}
{
\begin{definition}
Given a type derivation $\tdone$, the partial function \emph{computed by $\tdone$} 
is $\pfun{\tdone}:\CC^{2^n}\rightharpoonup\CC^{2^m}$
(where $n$ and $m$ are the input and output arity of $\tdone$) and is
defined by stipulating that $\pfun{\tdone}(\qrone)=\qrtwo$ iff
any initial state for $\qrone$ rewrites into a 
final state for $\qrthree$ and $\permone$, where $\qrthree=\permone^{-1}(\qrtwo)$. 
\end{definition}
}

Given a type derivation $\tdone$, $\pfun{\tdone}$ is either always undefined or
always defined. Indeed, the fact any initial configuration (for, say, $\qrone$)
rewrites to a final configuration or not does \emph{not} depend on $\qrone$ but only
on $\tdone$\condinc{.}{:
\begin{lemma}[Uniformity]
For every type derivation $\tdone$ and for every occurrences
$\occone_1,\ldots,\occone_n$, $\occtwo_1,\ldots,\occtwo_n$, there is a unitary
operator $\unopone$ such that
whenever $(\occone_1,\ldots,\occone_n,\qrone)\trans{\tdone}(\occtwo_1,\ldots,\occtwo_n,\qrtwo)$
it holds that $\qrtwo=\unopone(\qrone)$.
\end{lemma}
\begin{proof}
Observe that for every $\occone_1,\ldots,\occone_n$, $\occtwo_1,\ldots,\occtwo_n$ there
is \emph{at most} one of the rules defining $\trans{\tdone}$ which can be applied. Moreover,
notice that each rule acts uniformly on the underlying quantum register.
\end{proof}
}
In the following section, we will prove that $\pfun{\tdone}$ is always a total function, and
that it makes perfect sense from a quantum point of view.

\condinc{
{
\begin{example}\label{ex:runIAMQ} 
Consider the term $\termone_{\mathit{EPR}}$ and its type derivation $\tdone_{\mathit{EPR}}$ (see Example~\ref{ex:epr}).
\comment{Forgetting about terms and marking different occurrences of $\BB$ with distinct indices, we obtain:
$$
{\scriptsize
\urule
{
  \brule
  {\jd{\emcon}{}\BB_{9}\otimes\BB_{10}\linmap\BB_{11}\otimes\BB_{12}}
  {
    \brule
    {
      \brule
      {\jd{\emcon}{}:\BB_{21}\linmap\BB_{22}}
      {\jd{\BB_{23}}{\BB_{24}}}
      {\jd{\BB_{17}}{\BB_{18}}}
      {(\mathsf{E}_\linmap)}
    }
    {\jd{\BB_{19}}{\BB_{20}}}
    {\jd{\BB_{13},\BB_{14}}\BB_{15}\otimes\BB_{16}}
    {(\mathsf{I}_\otimes)}
  }
  {\jd{\BB_{5},\BB_{6}}{{}{}\BB_{7}\otimes\BB_{8}}}
  {(\mathsf{E}_\linmap)}
}
{
  \jd{\emcon}{}\BB_{1}\otimes\BB_{2}\linmap\BB_{3}\otimes\BB_{4}
}
{(\mathsf{I}_\linmap^2)}
}
$$
}
Let us consider the following computation of $\autom{\tdone_{\mathit{EPR}}}$: 
\begin{align*}
(\BB_1,&\BB_2,\qrone)\trans{\tdone}^*(\BB_{5},\BB_{6},\qrone)\trans{\tdone}^*(\BB_{13},\BB_{14},\qrone)
     \trans{\tdone}(\BB_{17},\BB_{19},\qrone)\trans{\tdone}^*(\BB_{23},\BB_{20},\qrone)\\
  &\trans{\tdone}^*(\BB_{24},\BB_{16})\trans{\tdone}(\BB_{24},\BB_{10},\qrone)\trans{\tdone}(\BB_{21},\BB_{10},\qrone)
      \trans{\tdone}(\BB_{22},\BB_{10},\mathbf{H}^1(\qrone))\\
  &\trans{\tdone}(\BB_{18},\BB_{10},\mathbf{H}^1(\qrone))\trans{\tdone}(\BB_{15},\BB_{10},\mathbf{H}^1(\qrone))
      \trans{\tdone}(\BB_{9},\BB_{10},\mathbf{H}^1(\qrone))\\
  &\trans{\tdone}(\BB_{11},\BB_{12},{\mathbf{CNOT}}^{1,2}(\mathbf{H}^1(\qrone)))\trans{\tdone}^*(\BB_{7},\BB_{8},\mathbf{CNOT}^{1,2}(\mathbf{H}^1(\qrone)))\\
  &\trans{\tdone}(\BB_{3},\BB_{4},\mathbf{CNOT}^{1,2}(\mathbf{H}^1(\qrone))).
\end{align*}
\noindent
Notice that thge occurrence of $\mathit{CNOT}$ acts as a synchronization operator: the second token is stuck in the occurrence $\BB_{10}$ 
until the first token arrives as a control input of the $\mathit{CNOT}$ and the corresponding reduction step actually occurs.  

\end{example}
}
}
{}

\section{Main Properties of $\IAM{\QL}$}\label{sect:mainres}
In this section, we will prove some crucial results about $\IAM{\QL}$. More specifically,
we prove that runs of this token machine are indeed finite and end in final states.
We proceed by putting $\QL$ in correspondence to $\MLL$, thus inheriting the same kind
of very elegant and powerful results enjoyed by \MLL\ token machines.

\subsection{A Correspondence Between $\MLL$ and $\QL$}\label{sec:mllql}
Any type derivation $\tdone$ can be put in correspondence with \emph{some} \MLL\ proofs.
We inductively define the map $\mapone{\cdot}$ from $\QL$ types to \MLL\ formulas as follows:
\condinc{$\mapone{\BB}=\alpha$, $\mapone{A\linmap B}=\lneg{{\mapone{A}}}\lpar\mapone{B}$,
$\mapone{A\otimes B}=\mapone{A}\otimes\mapone{B}$.}{
\begin{align*}
  \mapone{\BB}&=\alpha;\\
  \mapone{A\linmap B}&=\lneg{{\mapone{A}}}\lpar\mapone{B};\\
  \mapone{A\otimes B}&=\mapone{A}\otimes\mapone{B}.
\end{align*}}
Given a judgment $\judgone=\jd{\conone}{\termone:\typeone}$ and a natural number
$n\in\NN$, the \MLL\ sequent \emph{corresponding} to $\judgone$ and $n$ is the following one:
$$
\vdash\underbrace{\lneg{\alpha},\ldots,\lneg{\alpha}}_{\mbox{$n$ times}},
\lneg{(\mapone{\typetwo_1})},\ldots,\lneg{(\mapone{\typetwo_m})},\mapone{\typeone},
$$
where $\conone= x_1:\typetwo_1,\ldots,x_m:\typetwo_m$. For every $\tdone$, 
we define now a set of \MLL\ proofs $\maptwo{\tdone}$. This way,
every type derivation $\tdone$ for $\judgone=\jd{\conone}{\termone:\typeone}$ such that $n$ bits occur
in $\termone$, is put in relation to possibly many \MLL\ proofs of the sequent corresponding to $\judgone$ and $n$.
One among them is called the \emph{canonical proof} for $\tdone$. The set $\maptwo{\tdone}$ and
canonical proofs are defined by induction on the structure of the underlying
type derivation $\tdone$\condinc{. The multiplicative constructions of $\QL$ are mapped to the corresponding
$\MLL$ constructs, rules $\mathsf{(a_{q0})}$ and $\mathsf{(a_{q1})}$ are mapped to axioms, and rule $(\mathsf{a_U})$ is
mapped to a proof encoding a permutation of the involved atoms. For more details, please refer to~\cite{EV}. When the
latter is the identity, we get the \emph{canonical proof} for $\tdone$.
}
{:
\begin{varitemize}
\item
  If $\tdone$ is the type derivation
  $$
  \urule{}{\jd{\emcon}{\typet{\ket{0}}{\BB}}}{\mathsf{(a_{q0})}},
  $$
  then the only proof $\pmllone$ in $\maptwo{\tdone}$ is an atomic axiom.
  Similarly if the only rule in $\tdone$ is $\mathsf{(a_{q1})}$.
  Please notice that $\tdone$ contains \emph{one} bit, and as a consequence
  $\pmllone$ has the correct conclusion.
\item
  If $\tdone$ is
  $$
  \urule{}{\jd{\emcon}{U:\BB^{n}\linmap\BB^{n}}}{(\mathsf{a_U})},
  $$
  then $\tdone$ is in correspondence to all of the $n!$ possible
  cut-free proofs of the sequent 
  $$
  \vdash(\underbrace{\lneg{(\alpha\otimes\ldots\otimes\alpha)}}_{\mbox{$n$ times}}\lpar
  \underbrace{(\alpha\otimes\ldots\otimes\alpha)}_{\mbox{$n$ times}}
  $$
  obtained by starting from $n$ instances of an atomic axiom,
  gluing them together by the rule $\otimes$, and finally choosing
  one of the $n!$ possible permutations before applying $n$ times
  rule $\lpar$. The canonical proof is the one corresponding to the identity
  permutation.
\item
  If $\tdone$ is the type derivation
  $$
  \urule{}{\jd{\typet{x}{A}}{\typet{x}{A}}}{\mathsf{(a_v)}}
  $$
  then the only proof corresponding to $\tdone$ is the following
  $$
  \urule{}{\vdash{\lneg{{\mapone{A}}},\mapone{A}}}{\mathbf{Ax}}
  $$
\item
  If $\tdone$ is
  $$
  \urule
  {\tdtwo\pof\jd{\conone,\varone:\typeone}{\typet{\termone}{\typetwo}}}
  {\jd{\conone}{\typet{\abstr{\varone}{\termone}}{\typeone\linmap\typetwo}}}
  {(\mathsf{I}_\linmap^1)}
  $$
  where $\conone= x_1:\typeone_1,\ldots,x_m:\typeone_m$. 
  Then for all possible $\mll$ proof $\pmlltwo\in\maptwo{\tdtwo}$ of the $\mll$ sequent
  $$
  \smllone\;=\;\vdash\underbrace{\lneg{\alpha},\ldots,\lneg{\alpha}}_{\mbox{$n$ times}},
    \lneg{(\mapone{\typeone_1})},\ldots,\lneg{(\mapone{\typeone_m})},\lneg{(\mapone{\typeone})},\mapone{\typetwo}
  $$
  the following $\mll$ proof is in $\maptwo{\tdone}$:
  $$
  \urule
  {\pmlltwo\pof\smllone}
  {\vdash{\underbrace{\lneg{\alpha},\ldots,\lneg{\alpha}}_{\mbox{$n$ times}}},
    \lneg{(\mapone{\typeone_1})},\ldots,\lneg{(\mapone{\typeone_m})},\lneg{{\mapone{\typeone}}}\lpar{\mapone{\typetwo}}}
  {\mathbf{\lpar}}
  $$  
\item	
  If $\tdone$ is
  $$
  \urule
  {\tdtwo\pof{\jd{\conone,\varone:\typeone,\vartwo:\typetwo}{\termone:\typethree}}}
  {\jd{\conone}{\typet{\abstr{\langle\varone,\vartwo\rangle}{\termone}}
      {(\typeone\otimes\typetwo)\linmap\typethree}}}
  {(\mathsf{I}_\linmap^2)}
  $$
  where $\conone= z_1:\typefour_1,\ldots,z_m:\typefour_m, x:\typeone,y:\typetwo$, 
  then for all possible $\mll$ proofs $\pmlltwo\in\maptwo{\tdtwo}$ of the $\mll$ sequent
  $$
  \smllone\;=\;\vdash\underbrace{\lneg{{\alpha}},\ldots,\lneg{{\alpha}}}_{\mbox{$n$ times}},
    \lneg{{(\mapone{\typefour_1})}},\ldots,\lneg{{(\mapone{\typefour_m})}},
    \lneg{{\mapone{\typeone}}},\lneg{{\mapone{\typetwo}}},\mapone{\typethree}
  $$
  the following \MLL\ proof is in $\maptwo{\tdone}$:
  $$
  \urule{
    \urule
    {\pmlltwo\pof{\smllone}{}}
    {\vdash{\underbrace{\lneg{{\alpha}},\ldots,\lneg{{\alpha}}}_{\mbox{$n$ times}}},
      \lneg{{(\mapone{\typefour_1})}},\ldots,\lneg{{(\mapone{\typefour_m})}},(\lneg{{\mapone{\typeone}}}\lpar{\lneg{{\mapone{\typetwo}}}}),\mapone{\typethree}}
    {\mathbf{\lpar}}
  }
  {
  \vdash{\underbrace{\lneg{{\alpha}},\ldots,\lneg{{\alpha}}}_{\mbox{$n$ times}}},
    \lneg{{(\mapone{\typefour_1})}},\ldots,\lneg{{(\mapone{\typefour_m})}},(\lneg{{\mapone{\typeone}}}\lpar{\lneg{{\mapone{\typetwo}}}})\lpar\mapone{\typethree}
    }
  {\mathbf{\lpar}}
  $$   
\item
  If $\tdone$ is
  $$
  \brule
  {\tdtwo\pof{\jd{\conone}{\typet{M}{A\linmap B}}}{}}
  {\tdthree\pof{\jd{\contwo}{N: A}}{}}
  {\jd{\conone,\contwo}{MN:B}}
  {(\mathsf{E}_\linmap)}
  $$
  where $\conone= x_1:\typeone_1,\ldots,x_m:\typeone_m$  and  $\contwo= y_1:\typetwo_1,\ldots,y_k:\typetwo_k$
  then for all possible $\mll$ proofs $\pmllone\in\maptwo{\tdtwo}$ and $\pmlltwo\in\maptwo{\tdthree}$ of the $\mll$ sequents
  \begin{align*}
    \smlltwo\;=&\;
      \vdash\underbrace{\lneg{{\alpha}},\ldots,\lneg{{\alpha}}}_{\mbox{$n_1$ times}},
      \lneg{{(\mapone{\typeone_1})}},\ldots,\lneg{{(\mapone{\typeone_m})}},\lneg{{\mapone{\typeone}}}\lpar\mapone{B}\\
    \smllthree\;=&\;
      \vdash\underbrace{\lneg{{\alpha}},\ldots,\lneg{{\alpha}}}_{\mbox{$n_2$ times}},
      \lneg{{(\mapone{\typetwo_1})}},\ldots,\lneg{{(\mapone{\typetwo_k})}},\mapone{\typeone}
  \end{align*}
  the following $\mll$ proof is in $\maptwo{\tdone}$:
  $$
    \brule
    {
      \pmllone\pof\smlltwo
    }
    {
      \brule
      {\pmlltwo\pof\smllthree}
      {\urule{}{\vdash\lneg{{\mapone{\typetwo}}},\mapone{\typetwo}}{}}
      {
        \vdash\underbrace{\lneg{{\alpha}},\ldots,\lneg{{\alpha}}}_{\mbox{$n_2$ times}},
        \lneg{{(\mapone{\typetwo_1})}},\ldots,\lneg{{(\mapone{\typetwo_k})}},\mapone{\typeone}\otimes\lneg{{\mapone{\typetwo}}},\mapone{\typetwo}
      }
      {}
    }
    {
      \vdash\underbrace{\lneg{{\alpha}},\ldots,\lneg{{\alpha}}}_{\mbox{$n_1+n_2$ times}},\lneg{{(\mapone{\typeone_1})}},\ldots,\lneg{{(\mapone{\typeone_m})}},
      \lneg{{(\mapone{\typetwo_1})}},\ldots,\lneg{{(\mapone{\typetwo_k})}},\mapone{\typetwo}
    }
    {}
  $$
\item
  If $\tdone$ is
  $$
  \brule
  {\tdtwo\pof{\jd{\conone}{\typet{\termone}{\typeone}}}{}}
  {\tdthree\pof{\jd{\contwo}{\typet{\termtwo}{\typetwo}}}{}}
  {\jd{\conone,\contwo}{\typet{\termone\otimes\termtwo}{\typeone\otimes\typetwo}}}
  {(\mathsf{I}_\otimes)}
  $$
  where $\conone= x_1:\typeone_1,\ldots,x_m:\typeone_m$ and $\contwo= y_1:\typetwo_1,\ldots,y_k:\typetwo_k$,
  then for all possible $\mll$ proofs $\pmllone\in\maptwo{\tdtwo}$ and $\pmlltwo\in\maptwo{\tdthree}$ of the $\mll$ sequents
  \begin{align*}
    \smlltwo\;=&\;
      \vdash\underbrace{\lneg{{\alpha}},\ldots,\lneg{{\alpha}}}_{\mbox{$n_1$ times}},
      \lneg{{(\mapone{\typeone_1})}},\ldots,\lneg{{(\mapone{\typeone_m})}},\mapone{\typeone}\\
    \smllthree\;=&\;
      \vdash\underbrace{\lneg{{\alpha}},\ldots,\lneg{{\alpha}}}_{\mbox{$n_2$ times}},
      \lneg{{(\mapone{\typetwo_1})}},\ldots,\lneg{{(\mapone{\typetwo_k})}},\mapone{\typetwo}
  \end{align*}
  $\tdone$ is in correspondence to the $\mll$ proof
  $$
  \brule
  {\pmllone_1\pof{\smllone_1}{}}
  {\pmllone_2\pof{\smllone_2}{}}
  {\vdash{\underbrace{\lneg{{\alpha}},\ldots,\lneg{{\alpha}}}_{\mbox{$n_1+n_2$ times}},
      \lneg{{(\mapone{\typeone_1})}},\ldots,\lneg{{(\mapone{\typeone_m})}},
      \lneg{{(\mapone{\typetwo_1})}},\ldots,\lneg{{(\mapone{\typetwo_k})}},
      \lneg{{\mapone{\typeone}}}\otimes\lneg{{\mapone{\typetwo}}}}}
  {\mathbf{\otimes}}
  $$
\end{varitemize}
Observe how $\maptwo{\tdone}$ is a singleton whenever $\tdone$ does not contain any unitary operator of arity
(strictly) greater than $1$.}

Given an $\mll$ proof $\pmllone$, let us denote as $\statesmll{\pmllone}$ the class of all finite sequences of
atom occurrences in $\pmllone$.
The relation $\transmll{\pmllone}$ can be extended to a relation on $\statesmll{\pmllone}$ by stipulating that
\condinc{$}{$$}
(\occone_1,\ldots,\occone_{n-1},\occtwo,\occone_{n+1},\ldots,\occone_{m})
\transmll{\pmllone}
(\occone_1,\ldots,\occone_{n-1},\occthree,\occone_{n+1},\ldots,\occone_{m})
\condinc{$}{$$}
whenever $\occtwo\transmll{\pmllone}\occthree$. As usual, $\transmll{\pmllone}^+$ is the transitive closure
of $\transmll{\pmllone}$.

Let us now consider a type derivation $\tdone$ in $\QL$ and its quantum 
token machine $\autom{\tdone}$ and any $\pmllone\in\maptwo{\tdone}$. 
States of $\autom{\tdone}$ can be mapped to $\statesmll{\pmllone}$ by simply
forgetting the underlying quantum register and mapping any occurrence of
$\tdone$ to the corresponding atom occurrence in $\pmllone$. This way one gets a map
$\mapthree{\tdone}{\pmllone}{\cdot}:\states{\tdone}\rightarrow \statesmll{\pmllone}$
such that, given a state $\stgen=(\occone_1,\ldots,\occone_n,\qrone)$ in $\states{\tdone}$, 
$|\mapthree{\tdone}{\pmllone}{\stgen}|=n$, number of occurrences in $\stgen$ 
is the same as the length of $\mapthree{\tdone}{\pmllone}{\stgen}$.
Each reduction step on the token machine $\autom{\tdone}$ corresponds to \emph{at least one} 
reduction step in the $\mll$ machine $\mmllone{\pmllone}$, where $\pmllone\in\maptwo{\tdone}$ 
is the canonical proof:
\begin{lemma}\label{lemma:corresponding}
  Let us consider a token machine $\autom{\tdone}$  and two states $\stone,\sttwo\in\states{\tdone}$. If $\stone\trans{\tdone}\sttwo$
  and $\pmllone\in\maptwo{\tdone}$ is canonical, then $\mapthree{\tdone}{\pmllone}{\stone}\transmll{\pmllone}^+\mapthree{\tdone}{\pmllone}{\sttwo}$.
\end{lemma}
\begin{proof}
This goes by induction on the structure of $\tdone$.
\end{proof}
Any (possible) pathological situation on the quantum token machine, then, can be brought back to a corresponding 
(absurd) pathological situation in the $\MLL$ token machine. This is the principle that will guide us in the
rest of this section.
\condinc{\subsection{Termination, Progress, and Soundness}}{\subsection{Termination}}{}
The first property we want to be sure about is that every computation of any token machine 
$\autom{\tdone}$ always terminates. This is relatively simple to state and prove:
\begin{proposition}[Termination]~\label{prop:term} 
  Given a quantum token machine $\autom{\tdone}$, 
  any sequence $\stone\trans{\tdone}\sttwo\trans{\tdone}\ldots$ is finite.
\end{proposition}
\begin{proof}
  Suppose, for the sake of contradiction, than there exists an infinite computation in $\autom{\tdone}$. This implies by
  Lemma~\ref{lemma:corresponding} that there exists an infinite path in the token machine $\mmllone{\pmllone}$ where $\pmllone$ 
  is the canonical $\mll$ proof for $\tdone$. Absurd.\qed  
\end{proof}
\condinc{}{\subsection{Progress}}
Progress (i.e. deadlock-freedom) is more difficult to prove than termination. Again, however, we use in an
essential way the correspondence between $\QL$ and $\MLL$:
\begin{proposition}[Progress]\label{prop:progress}
Suppose $\tdone$ is a type derivation in $\QL$ and $\stone\in\states{\tdone}$ is initial. Moreover, suppose
that $\stone\trans{\tdone}^*\sttwo$. Then either $\sttwo$ is final or $\sttwo\trans{\tdone}\stthree$ for some
$\stthree\in\states{\tdone}$.
\end{proposition}
\condinc{The proof of Proposition~\ref{prop:progress} can be found in~\cite{EV}.}{

Given a type derivation $\tdone$, an \emph{argument occurrence} is any negative occurrence $(\typeone,\nconone)$ of $\BB$
in a $(\mathsf{a_U})$ axiom. We extend this definition to the corresponding atom occurrence when $\pmllone\in\maptwo{\tdone}$. 
A \emph{result occurrence} is defined similarly, but the occurrence has to be positive.
\begin{proof}
Let us consider a computation $\stgen_1\trans{\tdone}\ldots\trans{\tdone}\stgen_k$ on a 
quantum token machine $\autom{\tdone}$. Suppose that the state $\stgen_k$ is a deadlocked state, i.e. $\stgen_k$ 
is not a final state, and that there exists no $\stgen_{m}$ such that $\stgen_k\trans{\tdone}\stgen_m$.
The fact $\stgen_k$ is a deadlocked state means that $l\geq 1$ occurrences in $\stgen_k$ are argument occurrences, 
since the latter are the only points of synchronization of the machine. 
Let us consider any \emph{maximal} sequence 
\begin{equation}\label{equ:maxseq}
  \mapthree{\tdone}{\pmllone}{\stgen_1}\transmll{\pmllone}\ldots\transmll{\pmllone}\mapthree{\tdone}{\pmllone}{\stgen_k}
    \transmll{\pmllone}\mathsf{Q_1}\transmll{\pmllone}\ldots\transmll{\pmllone}\mathsf{Q}_n,
\end{equation}
where $\pmllone\in\maptwo{\tdone}$ is the canonical proof corresponding to $\tdone$.
Observe that in (\ref{equ:maxseq}), all occurrences of atoms in $\pmllone$ are visited exactly once, including those corresponding to argument and result
occurrences from $\tdone$. Notice, however, that the argument and result occurrences of the unitary operators affected by $\stgen_k$ cannot have been
visited along the subsequence $\mapthree{\tdone}{\pmllone}{\stgen_1}\transmll{\pmllone}\ldots\transmll{\pmllone}\mapthree{\tdone}{\pmllone}{\stgen_k}$
(otherwise we would visit the occurrences in $\stgen_k$ at least twice, which is not possible).
Now, form a directed graph whose nodes are the unitary constants $U_1,\ldots,U_h$ which block $\stgen_k$, plus a node $F$ (representing the conclusion of $\tdone$), 
and whose edges are defined as follows:
\begin{varitemize}
\item
  there is an edge from $U_i$ to $U_j$ iff along $\mathsf{Q_1}\transmll{\pmllone}\ldots\transmll{\pmllone}\mathsf{Q}_n$ one of the $l$ independent
  computations corresponding to a blocked occurrence in $\stgen_{k}$ is such that a result occurrence of $U_i$ is followed by an argument occurrence of 
  $U_j$ and the occurrences between them are neither argument nor result
  occurrences.
\item
  there is an edge from $U_i$ to $F$ iff along $\mathsf{Q_1}\transmll{\pmllone}\ldots\transmll{\pmllone}\mathsf{Q}_n$ one of the $l$ traces is such that
  a result occurrence of $U_i$ is followed by a final occurrence of an atom and the occurrences between them are neither argument nor result
  occurrences.
\end{varitemize}
The thus obtained graph has the following properties:
\begin{varitemize}
\item
  Every node $U_i$ has at least one incoming edge, because otherwise the configuration $\stgen_k$ would not be deadlocked.
\item
  As a consequence, the graph must be cyclic, because otherwise we could topologically sort it and get a node with no incoming edges (meaning
  that some of the $U_i$ would not be blocked!). Moreover, the cycle does not include $F$, because the latter only has incoming nodes.
\end{varitemize}
From any cycle involving the $U_j$, one can induce the presence of a cycle in the token machine $\mmllone{\pmlltwo}$ for some 
$\pmlltwo\in\maptwo{\tdone}$. Indeed, such a $\pmlltwo$ can be formed by simply choosing, for each $U_j$, the ``good'' permutation, 
namely the one linking the incoming edge and the outgoing edge which are part of the cycle. This way, we have reached the absurd starting from the 
existence of a deadlocked computation.
\qed
\end{proof}
The token machine $\autom{\tdone}$ can be built by following the structure of $\tdone$. However, the fact this gives rise to a well-behaved, unitary, function
requires proving some properties of $\autom{\tdone}$ (i.e. termination and progress) externally. One may wonder whether this could be avoided by 
taking a categorical approach and apply the so-called $\mathbf{Int}$-Construction~\cite{Joyal96} to the underlying category. This is not going to work, however, because
finite dimensional Hilbert spaces and unitary maps on them are not a \emph{traced} category. Of course, one could switch to linear maps, which indeed
turn Hilbert spaces into a traced category; one loses the strong link with quantum computation this way, however.
\subsection{Discussion}\label{sec:discussion}
}
\condinc{
The immediate consequence of the termination and progress results is that $\pfun{\tdone}$ is always a \emph{total}
function. The way $\autom{\tdone}$ is defined ensures that $\pfun{\tdone}$ is obtained by feeding some of the inputs of
a unitary operator $\unopone$ with some bits (namely those occurring in $\tdone$). $\unopone$ is itself obtained by composing the unitary operators
occurring in $\tdone$, which can thus be seen as a program computing a quantum circuit. In a way, then, token machines both show that $\QL$ is a truely quantum
calculus and can be seen as the right operational semantics for it.

But what is the relation between token machines and the equational theory on superposed type derivations introduced in Section~\ref{sect:equtheo}? It is easy
to extend the definition of $\pfun{\cdot}$ to superposed type derivations: if $\sutone=\sum_{i=1}^n\kappa_i\tdone_i$ then
$\pfun{\sutone}$ when fed with a vector $\vecone$ returns $\sum_{i=1}^n\kappa_i\pfun{\tdone_i}(\vecone)$. Remarkably, token machines behave
in accordance to the equational theory:
\begin{proposition}[Soundness]
If $\sutone\eqterm\suttwo$, then $\pfun{\sutone}=\pfun{\suttwo}$.
\end{proposition}
\vspace{-2ex}
\begin{proof}
This is an induction on the structure of a proof of $\sutone\eqterm\suttwo$. The base cases $\mathsf{beta}$ and $\mathsf{beta.pair}$ require
appropriate substitution lemmas.
\end{proof}

}
{
\newcommand{\qcirc}[1]{\langle #1\rangle}
The immediate consequence of the termination and progress results from Section~\ref{sect:mainres} is that $\pfun{\tdone}$ is always a \emph{total}
function. The way $\autom{\tdone}$ is defined ensures that $\pfun{\tdone}$ is obtained by feeding some of the input of
a unitary operator $\unopone$ with some bits (namely those occurring in $\tdone$). $\unopone$ is itself obtained by composing the unitary operators
occurring in $\tdone$, which can thus be seen as a program computing a quantum circuit, which we call $\qcirc{\tdone}$. Of course, $\pfun{\tdone}$
is nothing more than the function computed by  $\qcirc{\tdone}$. In a way, then, token machines both show that $\QL$ is a true quantum calculus 
and can be seen as the right operational semantics for it.
\begin{example}
Consider the term 
$
\termone_{\mathit{EPR}}=\abstr{\pair{\varone}{\vartwo}}{\mathit{CNOT}(\pairtens{\mathit{H}\varone}{y})}
$
and a type derivation $\tdone$ for it:
$$
\urule
{
  \brule
  {\jd{\emcon}{\mathit{CNOT}}:\BB\otimes\BB\linmap\BB\otimes\BB}
  {
    \brule
    {
      \brule
      {\jd{\emcon}{\mathit{H}}:\BB\linmap\BB}
      {\jd{\varone:\BB}{\varone:\BB}}
      {\jd{\varone:\BB}{\mathit{H}\varone:\BB}}
      {(\mathsf{E}_\linmap)}
    }
    {\jd{\vartwo:\BB}{\vartwo:\BB}}
    {\jd{\varone:\BB,\vartwo:\BB}{\pairtens{\mathit{H}\varone}{\vartwo}}:\BB\otimes\BB}
    {(\mathsf{I}_\otimes)}
  }
  {\jd{\varone:\BB,\vartwo:\BB}{\mathit{CNOT}(\pairtens{\mathit{H}\varone}{\vartwo}):\BB\otimes\BB}}
  {(\mathsf{E}_\linmap)}
}
{
  \jd{\emcon}{\termone_{\mathit{EPR}}}:\BB\otimes\BB\linmap\BB\otimes\BB
}
{(\mathsf{I}_\linmap^2)}
$$
Forgetting about terms and marking different occurrences of $\BB$ with distinct indices, we obtain:
$$
\urule
{
  \brule
  {\jd{\emcon}{}\BB_{9}\otimes\BB_{10}\linmap\BB_{11}\otimes\BB_{12}}
  {
    \brule
    {
      \brule
      {\jd{\emcon}{}:\BB_{21}\linmap\BB_{22}}
      {\jd{\BB_{23}}{\BB_{24}}}
      {\jd{\BB_{17}}{\BB_{18}}}
      {(\mathsf{E}_\linmap)}
    }
    {\jd{\BB_{19}}{\BB_{20}}}
    {\jd{\BB_{13},\BB_{14}}\BB_{15}\otimes\BB_{16}}
    {(\mathsf{I}_\otimes)}
  }
  {\jd{\BB_{5},\BB_{6}}{{}{}\BB_{7}\otimes\BB_{8}}}
  {(\mathsf{E}_\linmap)}
}
{
  \jd{\emcon}{}\BB_{1}\otimes\BB_{2}\linmap\BB_{3}\otimes\BB_{4}
}
{(\mathsf{I}_\linmap^2)}
$$
Now, consider the $\IAM{\QL}$ computation: 

\begin{align*}
(\BB_1,\BB_2,\qrone)&\trans{\tdone}^*(\BB_{5},\BB_{6},\qrone)\trans{\tdone}^*(\BB_{13},\BB_{14},\qrone)\\
  &\trans{\tdone}(\BB_{17},\BB_{19},\qrone)\trans{\tdone}^*(\BB_{23},\BB_{20},\qrone)\\
  &\trans{\tdone}(\BB_{24},\BB_{10},\qrone)\trans{\tdone}(\BB_{21},\BB_{10},\qrone)\\
  &\trans{\tdone}(\BB_{22},\BB_{10},\mathbf{H}^1(\qrone))\trans{\tdone}(\BB_{18},\BB_{10},\mathbf{H}^1(\qrone))\\
  &\trans{\tdone}(\BB_{15},\BB_{10},\mathbf{H}^1(\qrone))\trans{\tdone}(\BB_{9},\BB_{10},\mathbf{H}^1(\qrone))\\
  &\trans{\tdone}(\BB_{11},\BB_{12},{\mathbf{CNOT}}^{1,2}(\mathbf{H}^1(\qrone)))\trans{\tdone}^*(\BB_{7},\BB_{8},\mathbf{CNOT}^{1,2}(\mathbf{H}^1(\qrone)))\\
  &\trans{\tdone}(\BB_{3},\BB_{4},\mathbf{CNOT}^{1,2}(\mathbf{H}^1(\qrone))).
\end{align*}
Notice that $\mathit{CNOT}$ acts as a synchronization operator: the second token is stuck in the occurrence $\BB_{10}$ 
until the first token arrives as a control input of the $\mathit{CNOT}$ and the corresponding reduction step actually occurs.  
\comment{
The circuit $\qcirc{\tdone}$ is the following one:
\begin{align*}
  \Qcircuit @C=1em @R=1em {
    & \gate{H} & \ctrl{1} & \qw \\
    & \qw      & \targ    & \qw \\
  }
\end{align*}
}
\end{example}
\subsection{Soundness}\label{sec:sound}
What is the relation between token machines and the equational theory on superposed type derivations introduced in Section~\ref{sect:equtheo}? 

It is easy to extend the definition of $\pfun{\cdot}$ to superposed type derivations: if $\sutone=\sum_{i=1}^n\alpha_i\tdone_i$ then
$\pfun{\sutone}$ when fed with a vector $\vecone$ returns $\sum_{i=1}^n\alpha_i\pfun{\tdone_i}(\vecone)$. In the rest of 
this section, we will prove that token machines behave in accordance to the equational theory.

Suppose $\tdone$ is a type derivation for $\jd{\conone,\varone_1:\typeone_1,\ldots,\varone_m:\typeone_m}{\termone:\typetwo}$ and that,
for every $1\leq i\leq m$ there is a type derivation $\tdtwo_i$ for $\jd{\contwo_i}{\termtwo_i:\typeone_i}$. By induction
on the structure of $\tdone$, one can define a type derivation $\subst{\tdone}{\tdtwo_1,\ldots,\tdtwo_m}{\varone_1,\ldots,\varone_m}$ of
$\jd{\conone,\contwo_1,\ldots,\contwo_m}{\subst{\termone}{\termtwo_1,\ldots,\termtwo_m}{\varone_1,\ldots,\varone_m}:\typetwo}$ (see Lemma~\ref{lemma:substlemma}).
Moreover, from $\tdone,\tdtwo_1,\ldots,\tdtwo_m$ we can form a machine $\automc{\tdone}{\tdtwo_1,\ldots,\tdtwo_m}$ as follows:
\begin{varitemize}
\item
  The states of $\automc{\tdone}{\tdtwo_1,\ldots,\tdtwo_m}$ are in the form
  $(\occone_1,\ldots,\occone_n,\qrone)$ where:
  \begin{varitemize}
  \item
    $\occone_1,\ldots,\occone_n$ are occurrences of the 
    type $\BB$ in $\tdone,\tdtwo_1,\ldots,\tdtwo_m$; 
  \item
    $\qrone$ is a quantum register on $n$ qubits;
  \end{varitemize}
\item
  The transition function is itself obtained by taking the disjoint union of $\trans{\tdone},\trans{\tdtwo_1},\ldots,\trans{\tdtwo_n}$, plus
  \begin{varitemize}
  \item
    transitions of any positive occurrence of $\BB$ in $\typeone_i$ (in the conclusion of $\tdtwo_i$) to the corresponding
    occurrence of $\BB$ in $\typeone_i$ (this time in the conclusion of $\tdone$);
  \item
    transitions of any negative occurrence of $\BB$ in $\typeone_i$ (in the conclusione of $\tdone$) to the corresponding
    occurrence of $\BB$ in $\typeone_i$ (in the conclusion of $\tdtwo_i$).
  \end{varitemize}
\item
  Initial and final states are defined in the natural way, taking into account occurrences of $\BB$ in $\conone,\contwo_1,\ldots,\contwo_m,\typetwo$,
  but not those in $\typeone_1,\ldots\typeone_m$.
\end{varitemize}
The just defined machine is equivalent to the one built from the derivation $\subst{\tdone}{\tdtwo_1,\ldots,\tdtwo_n}{\varone_1,\ldots,\varone_m}$. 
This is stated by the following substitution lemma:
\begin{lemma}
Let $\tdone\pof\jd{\conone,\varone_1:\typeone_1,\ldots,\varone_m:\typeone_n}{\termone:\typetwo}$ and 
for every $1\leq i\leq m$ let $\tdtwo_i\pof\jd{\contwo_i}{\termtwo_i:\typeone_i}$. Then the automaton
$\autom{\subst{\tdone}{\varone_1,\ldots,\varone_n}{\tdtwo_1,\ldots,\tdtwo_n}}$ is equivalent to
$\automc{\tdone}{\tdtwo_1,\ldots,\tdtwo_n}$.
\end{lemma}
It is now possible to prove two key intermediate results towards soundness:
\begin{lemma}\label{lemma:1}
Let $\tdone\pof\jd{\conone}{(\abstr{\varone}{\termone})\termtwo:\typeone}$.
Then $\qcirc{\tdone}=\qcirc{\reduct{\tdone}}$.
\end{lemma}
\begin{lemma}\label{lemma:2}
Let $\tdone\pof\jd{\conone}{(\abstr{\pair{\varone}{\vartwo}}{\termone})(\pairtens{\termtwo}{\termthree}):\typeone}$.
Then $\qcirc{\tdone}=\qcirc{\reduct{\tdone}}$.
\end{lemma}
In order to prove Soundness Theorem, we need to introduce the following technical tool:
\begin{definition}[Superposed Quantum Circuits]
A \emph{superposed quantum circuits} of arity $(n,m)$ (where $n\leq m$) is a formal sums in the form
$$
\sum_{i=1}^n\alpha_i\qcone_i
$$
where $\alpha_i\in\CC$ and $\qcone_i$ is a quantum circuit on $m$ qubits of which $n$ are assigned a bit. 
\end{definition}
As an example, a superposed quantum circuit of arity $(2,4)$ looks as follows:

\begin{align*}
\alpha_1\cdot
\left(\hspace{22pt}
\begin{minipage}{5cm}
\Qcircuit @C=1em @R=.5em {
\lstick{\ket{b_1^1}} & \multigate{3}{C_1} & \qw \\
\lstick{\ket{b_2^1}} & \ghost{\mathcal{F}} & \qw \\
        & \ghost{\mathcal{F}} & \qw \\
        & \ghost{\mathcal{F}} & \qw \\
}
\end{minipage}
\right)
\hspace{10pt}
+
\hspace{10pt}
\alpha_2\cdot
\left(\hspace{22pt}
\begin{minipage}{5cm}
\Qcircuit @C=1em @R=.5em {
\lstick{\ket{b_1^2}} & \multigate{3}{C_2} & \qw \\
\lstick{\ket{b_2^2}} & \ghost{\mathcal{F}} & \qw \\
        & \ghost{\mathcal{F}} & \qw \\
        & \ghost{\mathcal{F}} & \qw \\
}
\end{minipage}
\right)
\end{align*}
Since  every type derivation $\pi$ computes a quantum circuit $\qcirc{\pi}$, every superposed type derivation $\sutone$ can be seen as a superposed quantum
circuit $\qcirc{\sutone}$. Moreover, the function $\pfun{\sum_{i=1}^{n}\alpha_i\qcone_i}$ 
computed by a superposed quantum circuit $\sum_{i=1}^n\alpha_i C_i$ can be defined similarly to what we have done for superposed
type derivations. Of course, $\pfun{\qcirc{\sutone}}=\pfun{\sutone}$. 

We now define the set of \emph{admissible circuit transformations}.

\begin{definition}[Admissible Transformations] Assume $\qcirc{\sutone}=\sum_{i=1}^n\alpha_i\qcone_i$ is a superposed quantum circuit. The following transformation are called \emph{admissible}:  
\begin{varenumerate}
\item
  One summand $\alpha\qcone_i$ is replaced by $\beta\qcone_i+\gamma\qcone_i$, where $\alpha=\beta+\gamma$;
\item
  One summand $\alpha\qcone_i$ where $\qcone_i$ has the following form
  \begin{align*}
  \Qcircuit @C=1em @R=.5em {
    \lstick{\ket{b_1}} & \qw \mbox{\vdots} & \multigate{1}{U} & \multigate{3}{D}  &\qw \\
    \lstick{\ket{b_m}} & \qw  & \ghost{U} & \ghost{D} & \qw \\
        & \qw & \qw & \ghost{D} & \qw \\
        & \qw & \qw & \ghost{D} & \qw \\
      }
  \end{align*}
  is replaced by a sum $\sum_{x\in B_m}\alpha\cdot\beta_x\cdot\qcone_x$
  where $B_m$ is the set of binary strings of length $m$, $\beta_x$ is the
  coefficient of $\ket{x}$ in $\unopone\ket{b_1\ldots b_m}$ and $\qcone_x$ is the 
  following circuit:
  \begin{align*}
  \Qcircuit @C=1em @R=.5em {
    \lstick{\ket{x_1}} & \qw \mbox{\vdots} & \multigate{3}{D}  &\qw \\
    \lstick{\ket{x_m}} & \qw  & \ghost{D} & \qw \\
        & \qw & \ghost{D} & \qw \\
        & \qw & \ghost{D} & \qw \\
      }
  \end{align*}
\end{varenumerate}
\end{definition}
Admissible transformations can be applied in both directions.
It is easy to prove that admissible transformations, when applied to 
a superposed circuit $\qcirc{\sutone}$, leave the underlying function unchanged. 
We are now ready to prove our soundness result:
\newcommand{\prone}{\mathsf{d}}
\newcommand{\prtwo}{\mathsf{e}}
\begin{theorem}[Soundness]\label{th:sound}
If $\sutone\eqterm\suttwo$, then $\pfun{\sutone}=\pfun{\suttwo}$.
\end{theorem}
\begin{proof}
Since $\pfun{\qcirc{\sutone}}=\pfun{\sutone}$, it is sufficient, by Proposition \ref{prop:normalform}, to show that, 
if $\sutone\eqtermnf\suttwo$, then $\qcirc{\suttwo}$ can be obtained from $\qcirc{\sutone}$ by iteratively applying 
one or more admissible transformations. This is an induction  on the structure of a proof  $\prone$ of $\sutone\eqtermnf\suttwo$. 
Let be $\mathsf{r}$ the last rule applied in $\prone$, where we enrich the thesis by stipluating that if the rules
in $\prone$ are all from $\axrls\cup\ccrls$, then $\sutone$ is a single type derivation and that going from 
$\qcirc{\sutone}$ to $\qcirc{\suttwo}$ can be done by performing \emph{at most one} admissible transformation
of the second kind. Some interesting cases:
\begin{varitemize}
 \item 
   $\mathsf{r}$ is $(\mathsf{beta.pair})$. The result follows by means of Lemma~\ref{lemma:1}.
 \item  
   $\mathsf{r}$ is $(\mathsf{beta})$. The result follows by means of Lemma~\ref{lemma:2}.
 \item 
   $\mathsf{r}$ is $(\mathsf{quant})$. Then $\prone$ is simply 
   $$\urule{\tdone\pof\jd{\emcon}{U\ket{\bitone_1\ldots\bitone_k}:\BB^k}}{\tdone\eqterm\unopone\ket{\bitone_1\ldots\bitone_k}}{\mathsf{quant}}$$
   and $\qcirc{\tdone}$ is simply the quantum circuit built  on  the unitary operator $U$, feeded with the input $\ket{\bitone_1\ldots\bitone_k}$.
   We know that $\unopone\ket{\bitone_1\ldots\bitone_k}$ is a superposed type derivation in the form $\suttwo=\sum_{x\in B_k}\alpha_x\tdone_x$, where
   $B_k$ is the set of all binary strings of length $k$ and $\tdone_x$ is the type derivation
   for $\ket{x}$ ($k$ applications of the rule ($\mathsf{I}_{\otimes})$ starting from the axioms for $\ket{\bitone_1}\ldots\ket{\bitone_k}$). 
   Such a derivation can be seen as the superposed quantum circuit of ariety $(k,k)$  $\qcirc{\suttwo}=\sum_{x\in B_k}\alpha_{x} {\ket{x}}$ 
   (where the binary string $\ket{x}$ can also seen as the trivial circuit that act on it as the identity)
   and  the  amplitudes $\alpha_{x}$ are exactly the coefficient of $\ket{x}$ in $U\ket{\bitone_1\ldots\bitone_k}$.  
   $\qcirc{\suttwo}$ can be plainly obtained from $\qcirc{\tdone}$ by means of the admissible transformation of the
   second kind by replacing the only summand $1\cdot C$ with the sum $\sum_{x\in B_k}1\cdot\alpha_{x} \ket{x}$.
 \item 
   $\mathsf{r}$ is a reflexive or a symmetric or a transitive closure. Trivial.
 \item
   $\mathsf{r}\in\ccrls$, then we know that $\sutone\eqtermnf\suttwo$ is derived from
   $\sutthree\eqtermnf\sutfour$, where $\sutthree$ is a \emph{single} type derivation and $\qcirc{\sutfour}$
   is obtained by applying either zero or one admissible transormations of the second kind to 
   $\qcirc{\sutthree}$. In other words, $\sutthree$ is
   \begin{align*}
  \Qcircuit @C=1em @R=.5em {
    \lstick{\ket{b_1}} & \qw \mbox{\vdots} & \multigate{1}{U} & \multigate{3}{D}  &\qw \\
    \lstick{\ket{b_m}} & \qw  & \ghost{U} & \ghost{D} & \qw \\
        & \qw & \qw & \ghost{D} & \qw \\
        & \qw & \qw & \ghost{D} & \qw \\
      }
  \end{align*}
  while $\sutfour$ is $\sum_{x\in B_m}\alpha\cdot\beta_x\cdot\qcone_x$
  where $B_m$ is the set of binary strings of length $m$, $\beta_x$ is the
  coefficient of $\ket{x}$ in $U\ket{b_1\ldots b_m}$ and $\qcone_x$ is the 
  following circuit:
  \begin{align*}
  \Qcircuit @C=1em @R=.5em {
    \lstick{\ket{x_1}} & \qw \mbox{\vdots} & \multigate{3}{D}  &\qw \\
    \lstick{\ket{x_m}} & \qw  & \ghost{D} & \qw \\
        & \qw & \ghost{D} & \qw \\
        & \qw & \ghost{D} & \qw \\
      }
  \end{align*}
   It is then clear that the effect of $\mathsf{r}$ to $\qcirc{\sutthree}$ consists in 
   modifying $D$, because $U$ cannot be affected. Moreover, the same modification is 
   perfomed by $r$ \emph{uniformly} on $D$ in any $\qcone_x$. We can then conclude 
   that there exists $E$ such that $\sutone$ is 
   \begin{align*}
  \Qcircuit @C=1em @R=.5em {
    \lstick{\ket{b_1}} & \qw \mbox{\vdots} & \multigate{1}{U} & \multigate{3}{E}  &\qw \\
    \lstick{\ket{b_m}} & \qw  & \ghost{U} & \ghost{E} & \qw \\
        & \qw & \qw & \ghost{E} & \qw \\
        & \qw & \qw & \ghost{E} & \qw \\
      }
  \end{align*}
  while $\suttwo$ is 
  \begin{align*}
  \Qcircuit @C=1em @R=.5em {
    \lstick{\ket{x_1}} & \qw \mbox{\vdots} & \multigate{3}{E}  &\qw \\
    \lstick{\ket{x_m}} & \qw  & \ghost{E} & \qw \\
        & \qw & \ghost{E} & \qw \\
        & \qw & \ghost{E} & \qw \\
      }
  \end{align*}

\end{varitemize}
This concludes the proof.
\end{proof}
}

\vspace{-2.3ex}
\section{Related Works}\label{sec:relw}
\condinc{
The role of GoI in quantum computing has already been explored in at least two works.
In~\cite{HH11}, a geometry of interaction model for Selinger and Valiron's quantum lambda calculus~\cite{SV06} is defined.
The model is formulated in particle-style.    
In~\cite{DLF11} \QMLL, an extension of \MLL\ with quantum modalities, is studied. \QMLL\ is 
sound and complete with respect to quantum circuits, and an interactive (particle-style) abstract machine is defined. 
The computational meaning of \QMLL\ proofs is given by means of their token machines: each cut-free \QMLL\ proof corresponds 
to an unique quantum circuit. In both cases, adopting a particle-style approach has a bad consequence: the ``quantum'' tensor 
product does \emph{not} coincide with the tensor product in the sense of linear logic. Here we show that adopting
the wave-style approach solves the problem.

Quantum extensions of game semantics are partially connected to this work. In~\cite{Del11} a game semantics 
for a simply-typed lambda calculus (similar to $\QL$) is introduced. The language uses a notion of extended variable, able to 
deal with tensor products. The game semantics is built around classical game semantics where, however, unitary quantum operations are the questions 
and measurements are the answers. A soundness result for the semantics is given.
A similar approach for a lambda calculus with quantum stores (i.e. in which quantum data are referred through pointers) has been explored in~\cite{DP08}.
Again, two tensor products are needed, unless one wants to drop the possibility of entangling qubits. 

Purely linear quantum lambda-calculi (\emph{with} measurements) can be given a fully abstract denotationl semantics, like the one proposed by
Selinger and Valiron~\cite{SelingerV08}. In their work, closure (necessary to interpret higher-order functions) is not obtained
via traces and is not directly related in any way to the geometry of interaction. Moreover, morphisms are just linear maps, 
and so the model is far from being an quantum operational semantics like the $\IAM{\QL}$.
}
{
The role of GoI in quantum computing has already been explored in at least two works.
In~\cite{HH11} a geometry of interaction model for Selinger and Valiron's quantum lambda calculus~\cite{SV06} is defined.
The model is formulated in particle-style.    
In~\cite{DLF11} \QMLL, an extension of \MLL\ with quantum modalities is studied. \QMLL\ is 
sound and complete with respect to quantum circuits, and an interactive, particle-style token machine is defined. 
The computational meaning of \QMLL\ proofs is given by means of the token machine: each cut-free \QMLL\ proof corresponds 
to an unique quantum circuit. In both cases, adopting a particle-style approach has a bad consequence: the ``quantum'' tensor 
product does \emph{not} coincide with the tensor product in the sense of linear logic. Here we show that adopting
the wave-style approach solves the problem.

Quantum extensions of game semantics are partially connected to our subject. In~\cite{Del11} a game semantics 
for a simply-typed lambda calculus (similar to $\QL$) is introduced. The language uses a notion of extended variable, able to 
deal with tensor products. The game semantics is built around classical game semantics where, however, quantum operations are the questions 
and measurements are the answers. A soundness result for the semantics is given.
A similar approach for a lambda calculus with quantum stores (i.e. in which quantum data are referred through pointers) has been explored in~\cite{DP08}.
Again, two tensor products are needed, unless one wants to drop the possibility of entangling qubits. 

Purely linear quantum lambda-calculi (\emph{with} measurements) can be given a fully abstract denotationl semantics, like the one proposed by
Selinger and Valiron~\cite{SelingerV08}. In their work, closure (necessary to interpret higher-order functions) is not obtained
via traces and is not directly related in any way to the geometry of interaction. Moreover, morphisms are just linear maps, 
and so the model is far from being an quantum operational semantics like the $\IAM{\QL}$.
}
\vspace{-2ex}

\section{Conclusions}\label{sect:conclusions}
\condinc{

\vspace{-2ex}

We have introduced $\IAM{\QL}$, an interactive abstract machine which provides a sound  operational 
characterization of any type derivation in a linear quantum $\lambda$-calculus $\QL$. This is the first example of
a concrete wave-style token machine whose runs cannot be seen simply as the asynchronous parallel composition
of particle-style runs. Interestingly, synchronization is intimately related to entanglement: if, for example,
only unary operators occur in a term (i.e. entanglement is \emph{not} possible), synchronization is not needed and 
everything collapses to the particle-style.

Our investigation is open to some possible future directions. A natural step will be to extend the syntax of terms 
and types with an exponential modality. The generalization of the wave-style token machine to this more 
expressive language would be an interesting and technically challenging subject. Something we see as relatively
easy is an extension of this framework to a calculus with measurements: token machines could cope with measurements
by evolving probabilistically~\cite{rairo2012}, while adapting the equational theory would probably be nontrivial.
Finally, giving a formal status to the connection between wave-style and the presence of entanglement is a fascinating 
subject which we definitely aim to investigate further. 
}
{
The definition of an elegant semantics is always a challenge in the case of  quantum functional languages.
This mainly holds for denotational models, but remains true also for operational, reduction-style semantics.
In this paper we introduce $\QL$, a linear quantum calculus  with explicit qubits, where quantum circuits can be easily encoded. 
This simple calculus is a good framework to further investigate the (deep) relationships between quantum computing and Girard's Geometry of Interaction.
We describe $\IAM{\QL}$, an interactive abstract machine which provides a sound  operational characterization of any $\QL$'s type derivation. 
$\QL$ quantum features force to move from the (usual) particle-style token machine model to the wave-style one, where different tokens 
circulate around a net (a type derivation) at the same time. 
Constants for n-ary unitary operators  act as \emph{synchronization} points: every token trips independently since it arrives at a unitary operator constant. 
In this case, computation takes place only if all input qubits occurrence has reached the unitary operator.
$\IAM{\QL}$ is a sound model: critical behaviors potentially introduced by the synchronization mechanism, can not happen in $\IAM{\QL}$ computations.
Our contribution can be summarized as follows:
\begin{varitemize}
\item
  The $\IAM{\QL}$ provides an elegant model for quantum programs written in $\QL$: each type derivation is interpreted as a 
  quantum circuit built on the set of quantum gates occurring in the underlying lambda-term;
\item
  we show that also wave-style token machines are sound with respect to an operational theory of superposed type derivations;
\item
  we give evidence that wave-style provides an original account of the quantum data entanglement phenomenon, since the notion 
  of synchronization we implicitly define is strongly connected to what happens to entangled data.
\end{varitemize}
Our investigation is open to some possible future directions. 
A natural step will be to extend the syntax of terms and type grammar with an exponential  modality.
The generalization of the wave-style token machine to this more 
expressive language would be an interesting and technically challenging subject. Something we see as relatively
easy is an extension of this framework to a calculus with measurements: token machines could cope with measurements
by evolving probabilistically\cite{rairo2012}, while adapting the equational theory would probably be nontrivial.
Finally, giving a formal status to the connection between wave-style and the presence of entanglement is a fascinating 
subject which we definitely aim to investigate further.
}
\condinc{

\bibliographystyle{abbrv}
{\scriptsize
\bibliography{biblio}
}
}
{
\bibliographystyle{abbrv}
\bibliography{biblio}
}
\end{document}